\documentclass[reqno,11pt]{amsart}

\usepackage[latin1]{inputenc}
\usepackage{amsfonts}
\usepackage{amssymb}

\usepackage{graphicx}         
\usepackage{tikz-cd}        
\usepackage{tikz}
\usetikzlibrary{positioning}

\usepackage{caption}       

\usepackage{mathtools}
\DeclarePairedDelimiter{\floor}{\lfloor}{\rfloor}

\newtheorem{thm}{Theorem}[section]
\newtheorem{cor}[thm]{Corollary}
\newtheorem{lemma}[thm]{Lemma}

\newtheorem{prop}[thm]{Proposition}
\theoremstyle{definition}
\newtheorem{defn}[thm]{Definition}
\theoremstyle{remark}

\newtheorem{remark}[thm]{Remark}
\newtheorem{example}[thm]{Example}

\renewcommand{\a}{\alpha}
\renewcommand{\b}{\beta}

\renewcommand{\l}{\lambda}
\newcommand{\s}{\sigma}

\newcommand{\Pb}[1]{\left\{\cdot\,,#1\right\}}
\newcommand{\pb}[1]{\left\{#1\right\}}

\newcommand{\lb}[1]{\[#1\]}

\renewcommand{\(}{\left(}
\renewcommand{\)}{\right)}
\renewcommand{\[}{\left[}
\renewcommand{\]}{\right]}

\newcommand{\set}[1]{\left\{#1\right\}}

\newcommand{\cL}{\mathcal L}
\newcommand{\cM}{\mathcal M}
\newcommand{\cN}{\mathcal N}
\newcommand{\cX}{\mathcal X}
\newcommand{\X}{ \cX}

\newcommand{\bbQ}{\mathbb Q}
\newcommand{\bbC}{\mathbb C}
\newcommand{\bbF}{\mathbb F}

\newcommand{\bbN}{\mathbb N}
\newcommand{\bbR}{\mathbb R}

\newcommand{\Sh}{\Delta}

\newcommand{\ul}[1]{\underline{#1}}
\newcommand{\ol}[1]{\overline{#1}}
\newcommand{\Aff}{\mathop{\rm{Aff}}\nolimits}
\newcommand{\Aut}{\mathop{\rm Aut}\nolimits}
\newcommand{\GL}{\mathop{{\rm GL}}\nolimits}
\newcommand{\GLp}{\mathop{\rm {GL}^+}\nolimits}
\newcommand{\Gr}{\mathop{\mathbf {Gr}}\nolimits}
\newcommand{\Grz}{\mathop{\mathbf {Gr_0}}\nolimits}

\renewcommand{\Im}{\mathop{\rm Im}\nolimits}
\newcommand{\LV}{\mathop{\rm LV}\nolimits}
\newcommand{\bLV}{\mathop{\mathbf {LV}}\nolimits}
\newcommand{\bLVz}{\mathop{\mathbf {LV_0}}\nolimits}
\newcommand{\leqs}{\leqslant}

\newcommand{\pp}[2]{\frac{\partial#1}{\partial#2}}

\newcommand{\we}{\wedge}
\newcommand{\Rk}{\hbox{Rk\,}}
\newcommand{\diff}{{\rm d }}

\newcommand{\Id}{\mathrm{Id}}
\newcommand{\KM}{\mathrm{KM}}
\newcommand{\B}[2]{\mathrm{B}(#1,#2)}
\newcommand{\wght}{\varpi}
\newcommand{\Span}{\mathop{\rm Span}}

\newif\ifprivate
\privatefalse

 \numberwithin{equation}{section}

\def\???{\ifprivate {\bf {???}} \marginpar{{\Huge {\bf ?}}}\else \fi}
\numberwithin{equation}{section}

\begin{document}

\nocite{*}

\title[Skew-symmetric Lotka-Volterra systems]{Morphisms and automorphisms of skew-symmetric Lotka-Volterra systems}

\author[C. A. Evripidou]{C. A. Evripidou}
\address{Charalampos Evripidou, Department of Mathematics and Statistics, University of Cyprus, P.O.~Box 20537,
  1678 Nicosia, Cyprus,\newline
   Department of Mathematics and Statistics,
		  La Trobe University, Melbourne,
		  Victoria 3086, Australia} \email{evripidou.charalambos@ucy.ac.cy}

\author[P. Kassotakis]{P. Kassotakis}
\address{Pavlos Kassotakis, Department of Mathematics and Statistics, University of Cyprus, P.O.~Box 20537, 1678
  Nicosia, Cyprus} \email{pavlos1978@gmail.com}

\author[P. Vanhaecke]{P. Vanhaecke}
\address{Pol Vanhaecke, Laboratoire de Math\'ematiques et Applications, UMR 7348 CNRS, Universit\'e de Poitiers, 11
  Boulevard Marie et Pierre Curie, Téléport 2 - BP 30179, 86 962 Chasseneuil Futuroscope Cedex,
  France}\email{pol.vanhaecke@math.univ-poitiers.fr}

\date{\today}


\begin{abstract}
We study the basic relation between skew-symmetric Lotka-Volterra systems and graphs, both at the level of objects
and morphisms, and derive a classification from it of skew-symmetric Lotka-Volterra systems in terms of graphs as
well as in terms of irreducible weighted graphs. We also obtain a description of their automorphism groups and of
the relations which exist between these groups. The central notion introduced and used is that of decloning of
graphs and of Lotka-Volterra systems. We also give a functorial interpretation of the results which we obtain.
\end{abstract}

\dedicatory{Dedicated to the Memory of our Friend and Teacher \\Pantelis A. Damianou, 1953 -- 2020}

\maketitle

\setcounter{tocdepth}{2}

\tableofcontents

\section{Introduction}
In its most general form, a Lotka-Volterra system in dimension $n$ is a dynamical system, described by the
following system of differential equations:
\begin{equation}\label{eq:LV_gen_intro}
  \dot x_i = \varepsilon_i x_i + \sum_{j=1}^n a_{i,j} x_i x_j, \ \ i=1,2, \dots , n \; .
\end{equation}
The coefficients $\varepsilon_i$ and $a_{i,j}$ are real or complex numbers, depending on whether one uses
$\bbF=\bbR$ or $\bbF=\bbC$ as the base field. These equations first appeared in the study of population dynamics
\cite{Lotka,Volterra}. The Lotka-Volterra systems which we study in this paper are \emph{skew-symmetric}, which
means that $\varepsilon_i=0$ and $a_{i,j}=-a_{j,i}$ for $1\leqslant i,j\leqslant n$. These conditions do not only
mean that we can interpret the coefficients $a_{i,j}$ as the entries of a skew-symmetric matrix $A$, but they also
imply that \eqref{eq:LV_gen_intro} is a Hamiltonian system: as Hamiltonian structure we can take the quadratic
Poisson structure $\pi_A$ on $\bbF^n$, defined in terms of the natural coordinates $x_1,\dots,x_n$ by the Poisson
brackets $\pb{x_i,x_j}_A:=a_{i,j}x_ix_j$ for $1\leqslant i,j\leqslant n$, and as Hamiltonian one can take the sum
of all coordinates, $H:=x_1+x_2+\cdots+x_n$. It is well-known that $\pi_A$ is indeed a Poisson structure (see for
example \cite[Section 8.2]{PLV}), and it is clear that the Hamiltonian vector field $\X_{H}:=\Pb{H}_A$ is given
by~\eqref{eq:LV_gen_intro} (with all $\varepsilon_i$ equal to zero). In what follows we refer simply to
skew-symmetric Lotka-Volterra systems as LV systems and we always take the Poisson structure $\pi_A$ as its
Hamiltonian structure.

It is natural to think of the skew-symmetric matrix $A=(a_{i,j})$ as being the adjacency matrix of a graph, having
an arc from the vertex $i$ to the vertex $j$ with value $a_{i,j}$ if $a_{i,j}\neq0$ and $i<j$, and having no arc
from $i$ to $j$ otherwise (see Definition \ref{def:graph} below for a more intrinsic description of these graphs,
which we will call \emph{skew-symmetric graphs}). Even if this natural relation between graphs and LV systems is
well-known \cite{Bog2,Damianou_graphs}, it has not been studied or exploited in the literature. The aim of the
present paper is to undertake the study of this relation, and of some of its consequences. As we will see,
it leads to a classification of LV systems in terms of the classification of skew-symmetric graphs, and we will use
the correspondence to determine the automorphism group of any LV system.

A first result (Proposition \ref{prp:functor}) states that morphisms between skew-symmetric graphs induce morphisms
between the corresponding LV systems. By the latter, we mean a smooth map which preserves both the Poisson
structure and the Hamiltonian. Moreover, the link between skew-symmetric graphs and LV systems is functorial, hence
isomorphic graphs lead to isomorphic LV systems and graph automorphisms lead to automorphisms of LV systems. By
construction, every such induced morphism is linear and a natural question is whether every linear morphism of LV
systems is induced by a graph morphism. The answer is negative in general, even for automorphisms, but there is an
easy characterization of the graphs for which every automorphism of the corresponding LV system is induced by a
graph automorphism. We call these graphs, and the corresponding LV systems, {irreducible}: an \emph{irreducible}
graph is characterized by the fact that no two of its vertices have identical neighborhoods; said differently, its
adjacency matrix has no equal rows (or columns, what amounts to the same).

For LV systems which are not irreducible, we show that the automorphism group is infinite, which proves that it is
strictly larger than the automorphism group of the underlying graph. In order to have a precise description of the
former automorphism group, we introduce the notion of decloning, which associates to a graph (and hence to its
associated LV system) an irreducible graph (and an irreducible LV system). At the level of the graph this is done
by identifying all vertices which have the same neighborhood; at the level of the LV system, the irreducible LV
system is obtained by a Poisson reduction, which is a morphism of Hamiltonian systems. It leads to the following
description of the automorphism group of the LV system, associated to any skew-symmetric graph $\Gamma$:
\begin{equation*}
  \Aut(\LV(\Gamma))\simeq\prod_{\ul s\in\ul S}\GLp(\wght_\Gamma({\ul s}),\bbF)\rtimes\Aut(\ul\Gamma,\wght_\Gamma)\;.
\end{equation*}
In this formula, $\ul\Gamma$ is the decloning of $\Gamma$ and the integers $\wght_\Gamma({\ul s})$ count for every
vertex of $\ul\Gamma$ how many vertices of $\Gamma$ have been identified by the decloning, in order to obtain
$\ul{s}$; finally, $\GLp(m,\bbF)$ stands for the subgroup of $\GL(m,\bbF)$, fixing some particular non-zero
vector. We also give a formula for the automorphism group of the graph $\Gamma$,
\begin{equation*}
  \Aut(\Gamma)\simeq\prod_{\ul s\in\ul S}\mathcal{S}_{\wght_\Gamma(\ul s)}\rtimes\Aut(\ul\Gamma,\wght_\Gamma)\;.
\end{equation*}%
As we will show, all these automorphism groups fit naturally in a commutative diagram, whose top line contains the
graphs and whose bottom line contains the corresponding LV systems.

We establish various properties of decloning. First, there is functoriality: we show that any surjective graph
morphism induces a morphism between the decloned graphs, and similarly at the level of their LV
systems, leading to the following commutative diagram of categories and functors:
\begin{equation*}
  \begin{tikzcd}[row sep=6ex, column sep=6ex]
    \Gr\arrow{r}{\rho}\arrow[swap]{d} {\LV}&\Grz\arrow{d}{\LV_0}\arrow[shift left=1ex] {l}{\imath_0}\\
    \bLV\arrow[swap]{r}{\sigma}&\bLVz\arrow[swap,shift right=1ex] {l}{\jmath_0}
  \end{tikzcd}
\end{equation*}
In this diagram, $\Gr$ and $\Gr_0$ stand for the categories of graphs, respectively irreducible graphs, with
surjective morphisms, and $\imath_0$ stands for the natural embedding of $\Gr_0$ in $\Gr$. Similarly, $\bLV$ and
$\bLVz$ stand for the categories of LV systems, respectively irreducible LV systems, with surjective morphisms, and
$\jmath_0$ stands for the natural embedding of $\bLVz$ in $\bLV$. The decloning functors are denoted by $\rho$ and
$\sigma$ and can be described as adjoint functors to $\imath_0$ and $\jmath_0$ respectively.

Decloning also plays an important role in the classification of LV systems. With some extra work, it follows from
what precedes that if two irreducible LV systems $\LV(\Gamma)$ and $\LV(\Gamma')$ are linearly isomorphic, then the
graphs $\Gamma$ and $\Gamma'$ are isomorphic; By decloning and thanks to a normal form that we give for linear
morphisms onto irreducible LV systems, we show that the statement remains true for arbitrary LV systems, not
necessarily irreducible ones. Moreover, we may replace in the statement ``linearly isomorphic'' by ``smoothly
isomorphic'', since we show that when two LV systems are smoothly isomorphic, then they are linearly isomorphic. In
conclusion, two LV systems are smoothly isomorphic if and only if the underlying graphs are
isomorphic. In that sense, the correspondence between graphs and LV systems is a bijective one.

Many LV systems have been shown to be integrable (i.e., Liouville integrable or superintegrable), often by using
Lax equations \cite{Bog1,Bog2,PPPP,RCD,suris_book,KKQTV}. We show that an LV system $\LV(\Gamma)$ is integrable if
and only if the decloned system $\LV(\ul\Gamma)$ is integrable. Moreover, we give a construction which provides a
(regular) Lax equation for $\LV(\Gamma)$ from a (regular) Lax equation for $\LV(\ul\Gamma)$, if any. This leads in
the particular case of the much studied Bogoyavlenskij systems $\B nk$, for which a particularly elegant Lax pair
is known \cite{Bog2}, to a Lax pair for any cloned Bogoyavlenskij system; for the latter systems, we will provide
also an alternative Lax pair.

We also give a population dynamics interpretation of cloning and decloning, showing that these operations are very
natural.

The structure of the paper is as follows. In Section \ref{sec:graphs} we introduce the basic graph terminology
which we will use, we define graph decloning and establish its functorial properties. We use it to determine the
group of graph automorphisms of a graph in terms of the graph automorphisms of its decloned graph. After recalling
the basic definitions of LV systems, associated to skew-symmetric graphs, we introduce and study in Section
\ref{sec:LV} the decloning of LV systems. Decloning of LV systems and the correspondence between graphs and LV
systems are studied from the functorial point of view, leading to a description of the automorphism groups of LV
systems and the classification of LV systems in terms of graphs. At the end of the section we give a population
dynamics interpretation of cloning and decloning. Section~\ref{sec:int_lax} is devoted to integrability and to the
construction of Lax equations for cloned LV systems, with special emphasis to the case of cloned Bogoyavlenskij
systems.

Throughout the paper, $\bbF=\bbR$ or $\bbF=\bbC$.
\section{Skew-symmetric graphs and their (auto-)morphisms}\label{sec:graphs}

In this section, we introduce in the case of graphs the basic construction of decloning which we will use to study
Lotka-Volterra systems and establish its properties. We also show how the automorphism groups of a graph and of its
decloned graph are related.

\subsection{Skew-symmetric and weighted graphs}

We  consider  two types of graphs, which we introduce first.
\begin{defn}\label{def:graph}
  A \emph{skew-symmetric graph} $\Gamma=(S,A)$ is a pair consisting of a finite set $S$, and a skew-symmetric map
  $A:S\times S\to\bbF$. A \emph{weighted graph} $(\Gamma,\wght)$ is a skew-symmetric graph $\Gamma=(S,A)$,
  equipped with a function $\wght:S\to\bbN^*:=\bbN\setminus\set0$.
\end{defn}

The finite set $S$ in the definition of a graph $\Gamma$ or weighted graph $(\Gamma,\wght)$ is called its
\emph{vertex~set}. The cardinality of~$S$ is denoted by $\vert\Gamma\vert$ and is called the \emph{order} of
$\Gamma$.  When $S=\set{1,2,\dots,n}$, the map $A$ is naturally thought of as a (skew-symmetric $n\times n$)
matrix; in general, we may also view $A$ as a matrix whose rows and columns are indexed by $S$. We therefore write
$a_{s,t}$ for $A(s,t)$, where $s,t\in S$ and we call $A$ the \emph{adjacency matrix} of~$\Gamma$. One may think of
$a_{s,t}$ as being the value of an arc from the vertex $s$ to the vertex~$t$. We call~$\wght$ the \emph{weight
  vector} of $\Gamma$ and we denote by $\vert\wght\vert:=\sum_{s\in{S}}\wght(s)$ the \emph{total weight} of $\wght$
and denote by $\wght_0$ the weight vector for which $\wght_0(s)=1$ for all $s\in S$.

The notion of a skew-symmetric graph, introduced above, is the skew-symmetric analog of the notion of a
non-oriented (valued) graph, since the latter can be defined by replacing in Definition \ref{def:graph}
\emph{skew-symmetric} by \emph{symmetric}; in particular, one naturally associates with any graph a skew-symmetric
graph by skew-symmetrizing its adjacency matrix. All results in this section are also valid for non-oriented
graphs, but in Sections \ref{sec:LV} and \ref{sec:int_lax} only skew-symmetric graphs and weighted graphs, as
defined above, will be relevant, hence we only consider skew-symmetric graphs.

It is often useful to represent a (possibly weighted) skew-symmetric graph by a picture, where the following
conventions turn out to be convenient: each vertex $s\in S$ is represented by a small circle containing $s$ as a
label, and for each pair of vertices $s,t$, a single arrow is drawn between them when $a_{s,t}\neq0$; we do not
differentiate between the two ways of doing this, as indicated in Figure \ref{fig:skew}:
\begin{figure}[h]
  \begin{tikzpicture}[->,shorten >=1pt,auto,node distance=1.5cm, thick,main node/.style={circle,draw,font=\bfseries}]
  \node[main node] (1) {$s$};
  \node[main node] (2) [right of=1] {$t$};
  \path
    (1) edge [edge label = $a_{s,t}$] (2);
  \end{tikzpicture}
  \qquad\hbox{or}\qquad
  \begin{tikzpicture}[->,shorten >=1pt,auto,node distance=1.5cm, thick,main node/.style={circle,draw,font=\bfseries}]
  \node[main node] (1) {$s$};
  \node[main node] (2) [right of=1] {$t$};
  \path
    (2) edge [edge label = $a_{t,s}$, swap] (1);
  \end{tikzpicture}
  \caption{When $a_{s,t}\neq0$ we draw one of the above arrows between $s$ and $t$, otherwise we draw no arrow
    between them.\label{fig:skew}}
\end{figure}
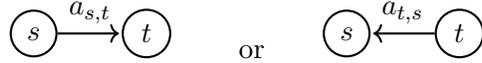

When $a_{s,t}=1$, which is often the case in our examples, we simply put an arrow from $s$ to $t$, omitting the
value $1$. For a weighted graph, we indicate the weight of a vertex as a red label, close to the vertex.

The above notions of graphs lead to the following natural definitions of their morphisms.
\begin{defn}\label{def:graph_morphism}
  Let $\Gamma=(S,A)$ and $\Gamma'=(S',A')$ be two skew-symmetric graphs. A \emph{graph morphism}
  $\Phi:\Gamma\to\Gamma'$ is a map $\Phi:S\to S'$ satisfying $a'_{\Phi(s),\Phi(t)}=a_{s,t}$ for all $s,t\in S$. Let
  $\wght$ and $\wght'$ be weight vectors for $\Gamma$, respectively for $\Gamma'$. A \emph{weighted graph morphism}
  $\Psi:(\Gamma,\wght)\to(\Gamma',\wght')$ is a graph morphism $\Psi:\Gamma\to\Gamma'$, satisfying
  $\wght'({\Psi(s)})\leqs\wght(s)$ for all $s\in S$.
\end{defn}

\goodbreak

The identity map of $\Gamma$ or of $(\Gamma,\wght)$ (i.e., of $S$) is denoted $\Id_{\Gamma}$. It is of course a
(weighted) graph morphism, just like the composition of two (weighted) graph
morphisms. Definition~\ref{def:graph_morphism} leads at once to the definitions of (weighted) graph isomorphisms
and (weighted) graph automorphisms; notice that when a graph morphism is bijective, its inverse is also a graph
morphism, making it into a graph isomorphism; for a weighted graph morphism $\Psi$, which is bijective, one needs
to ask in addition that it preserves the weight vectors, $\wght'({\Psi(s)})=\wght(s)$ for all $s\in S$, for it to
be an isomorphism. When two skew-symmetric graphs $\Gamma$ and $\Gamma'$ are isomorphic, we write
$\Gamma\simeq\Gamma'$ or $\Phi:\Gamma\simeq\Gamma'$, where $\Phi$ is an isomorphism, and the group of graph
automorphisms of $\Gamma$ is denoted by $\Aut(\Gamma)$; the same notations are used for weighted graphs
$(\Gamma,\wght)$. It is clear that $\Aut(\Gamma)$ is a finite group; it is a subgroup of the symmetric group
$\mathcal{S}_{\vert\Gamma\vert}$. For a weighted graph~$(\Gamma,\wght)$ it is clear that $\Aut(\Gamma,\wght)$ is a
subgroup of $\Aut(\Gamma)$.
\begin{remark}
It was shown in \cite{groups_graphs} that any finite group $G$ is the group of automorphisms of a (finite)
simple\footnote{$\Delta$ being simple means that it is non-oriented, without loops, and that all entries
of its adjacency matrix are 0 or 1.}  graph $\Delta=(S,E)$.  It follows from this result that $G$ is also the
group of automorphisms of a skew-symmetric graph~$\Gamma$; to construct $\Gamma$ it suffices to replace in $\Delta$
every edge by a vertex and two incident arrows, as indicated in Figure \ref{fig:groups_graphs}. Then $\Delta$ and
$\Gamma$ have the same automorphism groups, and the result follows.
\begin{figure}[h]
  \begin{tikzpicture}[-,shorten >=1pt,auto,node distance=1.5cm, thick,main node/.style={circle,draw,font=\bfseries}]
  \node[main node] (1) {$s$};
  \node[main node] (2) [right of=1] {$t$};
  \path
    (1) edge (2);
  \end{tikzpicture}
  \qquad\hbox{becomes}\qquad
  \begin{tikzpicture}[->,shorten >=1pt,auto,node distance=1.5cm, thick,main node/.style={circle,draw,font=\bfseries}]
  \node[main node] (1) {$s$};
  \node[main node] (2) [right of=1] {$u$};
  \node[main node] (3) [right of=2] {$t$};
  \path
  (2) edge (1)
  (2) edge (3);
  \end{tikzpicture}
  \caption{Every edge of the non-oriented graph $\Delta$ is replaced by a vertex and two arrows. A skew-symmetric
    graph is obtained, having the same automorphism group as $\Delta$.\label{fig:groups_graphs}}
\end{figure}
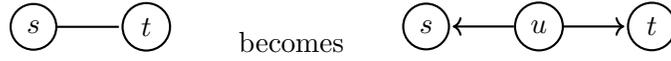

\end{remark}


\subsection{Cloning of weighted graphs}\label{par:cloning}
In order to motivate the definition of graph decloning, given in the next subsection, we first introduce the
operation of cloning. It associates with a weighted graph $(\Gamma,\wght)$, with $\Gamma=(S,A)$, a skew-symmetric
graph $\Gamma^\wght=(S^\wght,A^\wght)$. First, every vertex $s\in S$ gives rise to $\wght(s)$ vertices in
$S^\wght$, which we denote by $s_1,s_2,\dots,s_{\wght(s)}$ and which we call the \emph{clones} of the
vertex~$s$. So the constructed graph $\Gamma^\wght$ has $\vert{S^\wght}\vert=\vert\wght\vert$ vertices. Second, the
entries $a^\wght_{s_i,t_j}$ of the (skew-symmetric) adjacency matrix $A^\wght$ of $\Gamma^\wght$ are defined by
$a^\wght_{s_i,t_j}:=a_{s,t}$, for $s,t\in S$ and $1\leqs i\leqs\wght(s)$, $1\leqs j\leqs\wght(t)$. An example is
given in Figure \ref{fig:cloning} below.

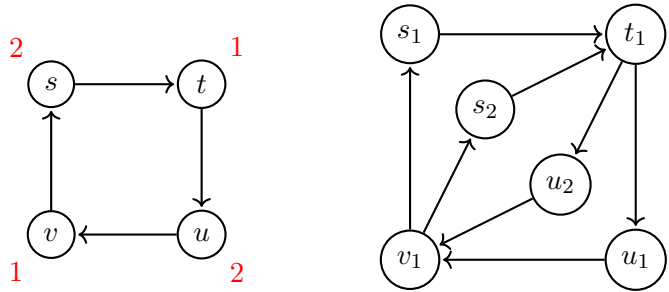
\begin{figure}[h]
  \begin{tikzpicture}[->,shorten >=1pt,auto,node distance=2cm,
                thick,main node/.style={circle,draw,font=\bfseries}]
  \node[main node] (1) [label=above left:$\color{red} 2$] {$s$};
  \node[main node] (2) [label=above right:$\color{red} 1$] [right of=1] {$t$};
  \node[main node] (3) [label=below right:$\color{red} 2$] [below of=2] {$u$};
  \node[main node] (4) [label=below left:$\color{red} 1$] [below of=1] {$v$};
  \path
    (1) edge (2)
    (2) edge (3)
    (3) edge (4)
    (4) edge (1);
\end{tikzpicture}
\qquad\qquad
\begin{tikzpicture}[->,shorten >=1pt,auto,node distance=1cm,
                thick,main node/.style={circle,draw,font=\bfseries}]
  \node[main node] (11) {$s_1$};
  \node (12) [right of=11] {};
  \node (13) [right of=12] {};
  \node[main node] (14) [right of=13] {$t_1$};
  \node (21) [below of=11] {};
  \node[main node] (22) [right of=21] {$s_2$};
  \node (23) [right of=22] {};
  \node (24) [right of=23] {};
  \node (31) [below of=21] {};
  \node (32) [right of=31] {};
  \node[main node] (33) [right of=32] {$u_2$};
  \node (34) [right of=33] {};
  \node[main node] (41) [below of=31] {$v_1$};
  \node (42) [right of=41] {};
  \node (43) [right of=42] {};
  \node[main node] (44) [right of=43] {$u_1$};
  \path
    (11) edge (14)
    (22) edge (14)
    (14) edge (44) edge (33)
    (33) edge (41)
    (44) edge (41)
    (41) edge (11) edge (22);
\end{tikzpicture}
\caption{On the left, a weighted graph $(\Gamma,\wght)$ with vertex set $\set{s,t,u,v}$ whose adjacency matrix
  takes values in $\set{-1,0,1}$. On the right, the associated skew-symmetric graph
  $\Gamma^\wght$.\label{fig:cloning}}
\end{figure}

Suppose that $S=\set{1,\dots,n}$, so that we can view $A$ as a matrix in the usual way.  Let us denote, for
$p,q\in\bbN^*$ by $\mathbf{1}_{p,q}$ the $p\times q$ matrix all of whose entries are equal to $1$. Then $A^\wght$
is obtained from $A$ by replacing each entry $a_{s,t}$ of $A$ by the matrix $a_{s,t}\mathbf1_{\wght(s),\wght(t)}$;
the labeling of the rows and columns of $A^\wght$ is given by $1_1,1_2,\dots,1_{\wght(1)},2_1,\dots,n_{\wght(n)}$,
in that order. It is clear that for any $s\in S$, the consecutive lines with labels $s_1,\,s_2,\dots s_{\wght(s)}$
of $A^\wght$ are identical, and similarly for the corresponding columns. It follows that~$A^\wght$ and $A$ have the
same rank. It also leads to the following definition:

\begin{defn}
  A skew-symmetric graph $\Gamma=(S,A)$ is said to be \emph{irreducible} if $A(s,\cdot)\neq A(t,\cdot)$ whenever
  $s\neq t$. Otherwise it is said to be \emph{reducible}. The same definition applies to weighted graphs.
\end{defn}
Said differently, $(S,A)$ is reducible when the adjacency matrix of $A$ has identical rows (or, equivalently,
identical columns). It is clear that this property is invariant under (weighted) graph isomorphisms. By the above,
if $(\Gamma,\wght)$ is any weighted graph with $\wght\neq\wght_0$, then $\Gamma^\wght$ is reducible.

\begin{remark}
To the cloned graph $\Gamma^\wght=(S^\wght,A^\wght)$ of a weighted graph $(\Gamma,\wght)$, with $\Gamma=(S,A)$, one
can add a weight vector $\wght':S^\wght\to\bbN^*$ and consider the cloned graph $(\Gamma^\wght)^{\wght'}$. It is easy
to see that the resulting graph is isomorphic to the cloned graph of $(\Gamma,\wght'')$, where the weight vector
$\wght'':S\to\bbN^*$ is defined by $\wght''(s):=\sum_{i=1}^{\wght(s)}\wght'(s_i)$ for all $s\in S$. It follows that
any repeated cloning of a weighted graph amounts to a single graph cloning. We will not elaborate on this fact,
because we are mainly interested in decloning, which is an idempotent operation, as we will see shortly.
\end{remark}

\subsection{Decloning of skew-symmetric graphs}\label{par:graph_decloning}

We now describe the inverse procedure, which we call \emph{decloning}. Let $\Gamma=(S,A)$ be a skew-symmetric
graph; we construct its \emph{decloned graph} $\ul{\Gamma}=(\ul S,\ul A)$ and its \emph{weighted decloned graph}
$(\ul{\Gamma},\wght_\Gamma)$. In a cloned graph the clones of a same vertex correspond to identical lines of the
adjacency matrix $A$, a property which can be used to partition a posteriori the vertex set $S$ into parts
containing the clones of each vertex. We therefore define an equivalence relation $\sim$ on $S$ by setting
$s\sim{t}$ if $a_{s,u}=a_{t,u}$ for all $u\in S$; said differently, if the lines (or, equivalently, the columns)
of~$A$ with labels $s$ and $t$ are identical. Let $\ul S:=S/\!\!\sim\,$ and denote by $p:S\to\ul S$ the canonical
projection map; for $s\in S$, we use the convenient notation $\ul s$ for $p(s)$. We get a well-defined
skew-symmetric map $\ul A:\ul S\times\ul S\to\bbF$ by setting $\ul a_{\ul s,\ul t}:=a_{s,t}$, for all $s,t\in
S$. Notice that the latter definition is equivalent to saying that $\ul A$ is the unique map making
$p:(S,A)\to(\ul{S},\ul{A})$ into a graph morphism. We call $p$ the \emph{decloning map} (of $\Gamma$). Finally, the
weight vector $\wght_\Gamma$ is defined for $\ul s\in\ul S$ by $\wght_\Gamma(\ul s):=\#\set{t\in S\mid s\sim t}$;
the decloning map can then also be viewed as a map (graph morphism) $p:\Gamma\to(\ul\Gamma,\wght_\Gamma)$. By
construction, all lines of $\ul A$ are different, so the decloned graph $\ul{\Gamma}$ is irreducible. When $\Gamma$
is irreducible, $\sim$ is trivial, so that $\Gamma\simeq\ul\Gamma$; as a consequence,
$\ul\Gamma\simeq\ul{\ul\Gamma}$ for any skew-symmetric graph $\Gamma$.  Cloning is the inverse procedure of
decloning in the sense that $\ul\Gamma^{\wght_\Gamma}\simeq\Gamma$ for any graph $\Gamma$; in particular, every graph
can be obtained by cloning a (unique) irreducible~graph. Notice that if $\wght$ is any weight vector for $\Gamma$,
then $\ul{\Gamma^\wght}\simeq\Gamma$ if and only if $\Gamma$ is irreducible.

Figure \ref{fig:cloning}, read from right to left, yields an example of decloning.

\subsection{Decloning of graph morphisms}

We show in this subsection that any \emph{surjective} graph morphism can be decloned, i.e., it induces a unique
(weighted) graph morphism between the (weighted) decloned graphs of the original skew-symmetric graphs. We start
with an elementary lemma, relating graph morphisms between two graphs and (weighted) graph morphisms between their
decloned graphs.
\begin{lemma}\label{lma:graph_morphisms}
  Let $\Gamma=(S,A)$ and $\Gamma'=(S',A')$ be skew-symmetric graphs, with decloning maps
  $p:\Gamma\to(\ul\Gamma,\wght_\Gamma)$ and $p':\Gamma'\to(\ul{\Gamma'},\wght_{\Gamma'})$. Suppose that
  $\Phi:\Gamma\to\Gamma'$ and $\Psi:\ul{\Gamma}\to\ul{\Gamma'}$ are two maps, making the following diagram
  commutative:
  \begin{equation}\label{dia:graphs_lma}
    \begin{tikzcd}[row sep=5ex, column sep=7ex]
      \Gamma\arrow{r} {\Phi}\arrow[swap]{d} {p}&\Gamma'\arrow{d}{p'}\\
            (\ul\Gamma,\wght_\Gamma)\arrow[swap]{r}{\Psi}&(\ul{\Gamma'},\wght_{\Gamma'})
    \end{tikzcd}
  \end{equation}
  Then,
  \begin{enumerate}
    \item $\Phi$ is a graph morphism if and only if $\Psi$ is a graph morphism.
    \item If $\Phi$ is surjective, then $\Phi$ is a graph morphism if and only if $\Psi$ is a weighted graph
      morphism. In this case, if $s,t\in S$, then $s\sim t$ if and only if $\Phi(s)\sim\Phi(t)$.
  \end{enumerate}
\end{lemma}
\begin{proof}
First, notice that the commutativity of the diagram can be written as $\Psi(\ul s)=\ul{\Phi(s)}$, for all
$s\in{}S$.  The conditions that $\Phi$, respectively $\Psi$, is a graph morphism mean that, for all $s,t\in S$,
\begin{equation*}
  a'_{\Phi(s),\Phi(t)}=a_{s,t}\;,\quad\text{resp.}\quad \ul a'_{\Psi(\ul s),\Psi(\ul t)}=\ul a_{\ul s,\ul t}\;.
\end{equation*}%
By the definition of decloning, $\ul a_{\ul s,\ul t}=a_{s,t}$ and
$\ul{a'}_{\ul{\Phi(s)},\ul{\Phi(t)}}=a'_{{\Phi(s),\Phi(t)}}$. So the equivalence in item (1) (and the inverse
implication of the first equivalence in item (2)) follows from the commutativity of the diagram. In order to prove
the remaining statements in item (2), suppose that $\Phi$ is a surjective graph morphism.  For $s,t\in S$, by
definition, $s\sim t$ iff $a_{s,u}=a_{t,u}$ for all $u\in S$. Since~$\Phi$ is a graph morphism, this is equivalent
to $a'_{\Phi(s),\Phi(u)}=a'_{\Phi(t),\Phi(u)}$ for all \hbox{$u\in S$}. Since $\Phi$ is surjective, this means that
$a'_{\Phi(s),v}=a'_{\Phi(t),v}$ for all $v\in S'$, i.e., $\Phi(s)\sim\Phi(t)$. This shows the second statement in
item (2); notice that the existence of $\Psi$ was not used to prove this equivalence. It follows, using the
commutativity of the diagram and the surjectivity of $\Phi$, that
\begin{align*}
  \wght_{\Gamma'}(\Psi(\ul s))&=\wght_{\Gamma'}(\ul{\Phi(s)})=\#\set{v\in S'\mid\Phi(s)\sim v}\\
    &\leqs\#\set{t\in S\mid\Phi(s)\sim \Phi(t)}=\#\set{t\in S\mid s\sim t}=\wght_\Gamma(\ul s)\;,
\end{align*}%
for all $s\in S$, so that $\Psi$ is a weighted graph morphism, which completes the proof of item (2).
\end{proof}
Notice that in the lemma, given $\Phi$, a map $\Psi$ making the diagram commutative is unique, since $p$ is
surjective.
\begin{prop}\label{prp:graph_decloning}
  Suppose that $\Gamma$ and $\Gamma'$ are two skew-symmetric graphs. As above, denote by $(\ul\Gamma,\wght_\Gamma)$
  and $(\ul\Gamma',\wght_{\Gamma'})$ the weighted decloned graphs of $\Gamma$ and $\Gamma'$ and by $p$ and $p'$ the
  decloning maps.
  \begin{enumerate}
    \item If $\Phi:\Gamma\to\Gamma'$ is a surjective graph morphism, then $\Phi$ induces a unique graph morphism
      $\ul\Phi:\ul\Gamma\to\ul{\Gamma'}$ such that the following diagram of (surjective) graph morphisms is
      commutative:
        \begin{equation}\label{dia:graphs}
          \begin{tikzcd}[row sep=5ex, column sep=7ex]
            \Gamma\arrow{r} {\Phi}\arrow[swap]{d} {p}&\Gamma'\arrow{d}{p'}\\
            (\ul\Gamma,\wght_\Gamma)\arrow[swap]{r}{\ul\Phi}&(\ul{\Gamma'},\wght_{\Gamma'})
          \end{tikzcd}
        \end{equation}
        Moreover, $\ul\Phi:(\ul\Gamma,\wght_\Gamma)\to(\ul{\Gamma'},\wght_{\Gamma'})$ is a weighted graph morphism.
  \item If $\Gamma''$ is another skew-symmetric graph and $\Phi':\Gamma'\to\Gamma''$ is another surjective graph
    morphism, then $\ul{\Phi'\circ\Phi}=\ul{\Phi'}\circ\ul\Phi$; also,
    $\ul{\Id_\Gamma}=\Id_{(\ul\Gamma,\wght_\Gamma)}$.
  \item Suppose that $\Psi:(\ul\Gamma,\wght_\Gamma)\to(\ul{\Gamma'},\wght_{\Gamma'})$ is a weighted graph morphism,
    which is surjective. Then there exists a surjective graph morphism $\Phi:\Gamma\to\Gamma'$ such that
    $\ul\Phi=\Psi$.
  \end{enumerate}
\end{prop}
\begin{proof}
We write $\Gamma=(S,A)$, as before.  We first prove item (1). As we pointed out in the proof of Lemma
\ref{lma:graph_morphisms}, the fact that $\Phi$ is a surjective graph morphism implies that for any $s,t\in S$,
$s\sim t$ iff $\Phi(s)\sim\Phi(t).$ It follows that~$\ul\Phi$ is well-defined by $\ul\Phi(\ul s):=\ul{\Phi(s)}$ and
$\ul\Phi$ is the unique map making (\ref{dia:graphs}) commutative. In view Lemma \ref{lma:graph_morphisms}, this
shows item (1). The uniqueness of~$\ul\Phi$ in item (1) implies at once item (2). Suppose now that
$\Psi:(\ul\Gamma,\wght_\Gamma)\to(\ul{\Gamma'},\wght_{\Gamma'})$ is a surjective weighted graph morphism. Since
$\wght_{\Gamma'}(\Psi(\ul s))\leqslant\wght_\Gamma(\ul s)$ for all $\ul s\in \ul S,$ one can map surjectively the
set of clones of any $\ul s\in \ul S$ to the clones of $\Psi(\ul s)$; doing this for all $\ul s\in \ul S$ one gets
a surjective map $\Phi$ such that $p'\circ\Phi=\Psi\circ p$. By item (1) in Lemma \ref{lma:graph_morphisms}, $\Phi$
is a (surjective) graph morphism; by uniqueness,~$\ul\Phi=\Psi$.
\end{proof}
The unique (weighted) graph morphism $\ul\Phi$, induced by $\Phi$ is called its \emph{(weighted) decloned graph
morphism}.  Property (2) in the proposition says that decloning is a functor. We will come back to this in
Section \ref{par:cat_graphs}.

In general, Proposition \ref{prp:graph_decloning} is false for a graph morphism $\Phi$ which is not surjective. A
counterexample is given in Figure \ref{fig:counter_ex}.
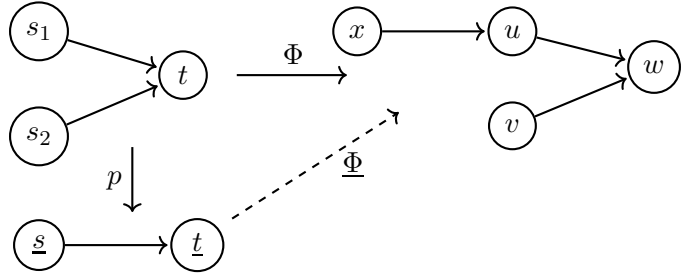
\begin{figure}[h]
  \begin{tikzpicture}[->,shorten >=1pt,auto,thick,main node/.style={circle,draw,font=\bfseries}]
  \node[main node] (s1) {$s_1$};
  \node[main node] (s2) [below=.6cm of s1] {$s_2$};
  \node[main node] (t) [above right=.3cm and 1.4cm of s2] {$t$};
  \node (r1) [right=.1cm of t] {};
  \node (r2) [right=.7 of s2] {};

  \node[main node] (x) [right=3.5cm of s1] {$x$};
  \node[main node] (u) [right=1.4 of x] {$u$};
  \node[main node] (v) [below=.6cm of u] {$v$};
  \node[main node] (w) [above right=.3cm and 1.4cm of v] {$w$};
  \node (r3) [right=1.5cm of r1] {};
  \node (r4) [below right=.1cm and .4cm of r3] {};

  \node[main node] (s) [below=2.1cm of s1] {$\ul s$};
  \node[main node] (t0) [right=1.4cm of s] {$\ul t$};
  \node (r5) [below=.9cm of r2] {};
  \node (r6) [right=.9cm of r5] {};

  \path
  (s1) edge (t)
  (s2) edge (t)

  (x) edge (u)
  (u) edge (w)
  (v) edge (w)

  (s) edge (t0)

  (r1) edge node[above] {$\Phi$} (r3)
  (r2) edge node[left] {$p$} (r5)
  (r6) edge [dashed] node[right,xshift=.2cm] {$\underline{\Phi}$} (r4);
\end{tikzpicture}

\caption{The reducible graph on the left is embedded in the irreducible graph on the right. Any map $\ul\Phi$
  making the diagram commutative must send $\ul s$ both to $u$ and $v$, which is impossible. The lack of
  surjectivity of $\Phi$ is responsible for this.\label{fig:counter_ex}}
\end{figure}

\subsection{Graph isomorphisms and automorphisms}

We use Proposition~\ref{prp:graph_decloning} to give a description of the graph automorphism group $\Aut(\Gamma)$
of a skew-symmetric graph in terms of the automorphism group $\Aut(\ul\Gamma,\wght_\Gamma)$ of its weighted
decloned graph $(\ul\Gamma,\wght_\Gamma)$. Recall that $\Gamma$ is always supposed to be finite, so these groups
are finite groups.  We first consider graph isomorphisms.
\begin{prop}\label{prp:graph_autom_decloning}
  Let $\Gamma$ and $\Gamma'$ be two skew-symmetric graphs.
  \begin{enumerate}
    \item Suppose that $\Phi:\Gamma\to\Gamma'$ is a surjective graph morphism, with weighted decloned graph
      morphism $\ul\Phi:(\ul\Gamma,\wght_\Gamma)\to(\ul{\Gamma'},\wght_{\Gamma'})$. Then $\ul\Phi$ is a weighted graph
      isomorphism if and only if $\Phi$ is a graph isomorphism.
    \item $\Gamma\simeq\Gamma'$ if and only if $(\ul\Gamma,\wght_\Gamma)\simeq(\ul{\Gamma'},\wght_{\Gamma'})$.
  \end{enumerate}
\end{prop}
\begin{proof}
It follows from item (2) in Proposition \ref{prp:graph_decloning} that if $\Phi$ is a graph isomorphism, then
$\ul\Phi:\ul\Gamma\to\ul{\Gamma'}$ is a graph isomorphism and
$\ul\Phi:(\ul\Gamma,\wght_\Gamma)\to(\ul{\Gamma'},\wght_{\Gamma'})$ is a weighted graph isomorphism. Conversely,
suppose that the induced map $\ul\Phi:(\ul\Gamma,\wght_\Gamma)\to(\ul{\Gamma'},\wght_{\Gamma'})$ is a weighted
graph isomorphism. Then, for all $\ul s\in\ul S$, $\wght_{\Gamma'}({\ul\Phi(\ul s)})=\wght_\Gamma(\ul s)$, so that
$\vert\Gamma\vert=\vert\wght_\Gamma\vert=\vert\wght_{\Gamma'}\vert=\vert\Gamma'\vert$. Since our graphs are finite
and since $\Phi$ is assumed surjective, this shows that $\Phi$ is bijective, hence is a graph isomorphism. It
proves item (1) and also the direct implication of item (2). In order to prove the inverse implication of item (2),
suppose that $\Psi:(\ul\Gamma,\wght_\Gamma)\simeq(\ul{\Gamma'},\wght_{\Gamma'})$ is a weighted graph
isomorphism. Let $\Phi:\Gamma\to\Gamma'$ be a surjective graph morphism such that $\ul\Phi=\Psi$, as provided by
item (3) in Proposition \ref{prp:graph_decloning}. Since $\Psi$ is bijective and weight-preserving,
\begin{equation*}
  \vert\Gamma'\vert=\sum_{\ul u\in\ul{S'}}\wght_{\Gamma'}(\ul u)=\sum_{\ul s\in\ul S}\wght_{\Gamma'}(\Psi(\ul s))
  =\sum_{\ul s\in\ul S}\wght_\Gamma(\ul s)=\vert\Gamma\vert\;,
\end{equation*}%
where $\ul S$ and $\ul{S'}$ stand for the vertex sets of $\ul{\Gamma}$ and $\ul{\Gamma'}$, as before. It follows
that the surjective graph morphism $\Phi$ is a bijection, hence a graph isomorphism, and $\Gamma\simeq\Gamma'$.
\end{proof}

\begin{prop}\label{prp:graph_autom}
  Let $\Gamma=(S,A)$ be a skew-symmetric graph, whose weight\-ed decloned graph is denoted by
  $(\ul\Gamma,\wght_\Gamma)$, where $\ul\Gamma=(\ul S,\ul A)$. Then
  \begin{equation*}
    0\to\prod_{\ul s\in\ul S}\mathcal{S}_{\wght_\Gamma(\ul s)}\to\Aut(\Gamma)\to\Aut(\ul\Gamma,\wght_\Gamma)\to0
  \end{equation*}%
  is a split short exact sequence of groups. As a consequence, $\Aut(\Gamma)$ is a semi-direct product,
  \begin{equation*}
    \Aut(\Gamma)\simeq\prod_{\ul s\in\ul S}\mathcal{S}_{\wght_\Gamma(\ul s)}\rtimes\Aut(\ul\Gamma,\wght_\Gamma)\;.
  \end{equation*}%
\end{prop}
\begin{proof}
Let $s\in S$ and consider the equivalence class $\ul s$ of $s$, which we view here as a subset of $S$, of
cardinality $\wght_\Gamma(\ul s)$. Let $\varsigma$ be any permutation of~$\ul s$, $\varsigma\in\mathcal
S_{\ul{s}}\simeq\mathcal S_{\wght_\Gamma(\ul s)}$ and let $\Phi$ be the extension of $\varsigma$ to a permutation
of $S$ which fixes all vertices of $S\setminus\set{\ul s}$. According to Lemma \ref{lma:graph_morphisms}, with
$\Gamma=\Gamma'$ and $\Psi=\Id_{\ul\Gamma}$, $\Phi$ is a graph morphism, hence a graph automorphism. Since this can
be done for any $\ul s\in\ul S$, and since for $\ul s\neq\ul t\in\ul S$, the permutations of $S$ corresponding to
$\ul s$ and $\ul t$ have disjoint support, these permutations generate a group of graph automorphisms of $\Gamma$,
isomorphic to $\prod_{\ul s\in\ul S} \mathcal{S}_{\wght_\Gamma(\ul s)}$; this accounts for its inclusion in
$\Aut(\Gamma)$. The surjectivity of the group homomorphism
$\Aut(\Gamma)\to\Aut(\ul\Gamma,\wght_\Gamma):\Phi\mapsto\ul\Phi$ is proven in the same way as item (2) in
Proposition~\ref{prp:graph_autom_decloning}. Since the neutral element of $\Aut(\ul\Gamma,\wght_\Gamma)$ is
$\Id_{\ul\Gamma}$, every graph automorphism in the kernel of this morphism can only permute equivalent vertices of
$\Gamma$, hence is the graph automorphism corresponding to an element of
$\prod_{\ul{s}\in\ul{S}}\mathcal{S}_{\wght_\Gamma(\ul s)}$; conversely, the graph automorphism corresponding to any
element of $\prod_{\ul{s}\in\ul{S}}\mathcal{S}_{\wght_\Gamma(\ul s)}$ is in the kernel of the group homomorphism
$\Aut(\Gamma)\to\Aut(\ul\Gamma,\wght_\Gamma)$. This shows that the short sequence is exact. To show that it is
split, we need to construct a section of the surjection $\Aut(\Gamma)\to\Aut(\ul\Gamma,\wght_\Gamma)$. To do this,
we fix a numbering of the elements of $\ul s$, for every $\ul s\in\ul S$, writing them as
$s_1,s_2,\dots,s_{\wght_\Gamma(\ul s)}$. Given $\Psi\in\Aut(\ul\Gamma,\wght_\Gamma)$ we get a bijection $\ol\Psi$
of $\Gamma$ by setting $\ol\Psi(s_i):=\Psi(\ul s)_i$ for $\ul s\in \ul S$ and $1\leqs i\leqs\wght_\Gamma(\ul
s)$. According to item (1) in Lemma \ref{lma:graph_morphisms}, $\ol\Psi$ is a graph (auto-)morphism. Clearly, when
$\Psi_1$ and $\Psi_2$ are two elements of $\Aut(\ul\Gamma,\wght_\Gamma)$, then
$\ol{\Psi_1\circ\Psi_2}=\ol{\Psi_1}\circ\ol{\Psi_2}$, so that the map $\Psi\mapsto\ol\Psi$ is a group homomorphism;
it is a section because $\ul{(\ol\Psi)}=\Psi$ for any $\Psi\in\Aut(\ul\Gamma,\wght_\Gamma)$, as follows from the
definitions of $\ol\Psi$ and of decloning of surjective graph morphisms.
\end{proof}
\subsection{Functorial interpretation}\label{par:cat_graphs}
We finish this section by giving a functorial interpretation\footnote{We only use the basics of category theory;
  see for example the book \cite{cats} which is freely available online.} of the above results, in particular of
Propositions \ref{prp:graph_decloning} and \ref{prp:graph_autom_decloning}. Let us denote by $\Gr$ the category of
skew-symmetric graphs (with values in $\bbF$) with surjective graph morphisms and by $\Grz$ the subcategory of
irreducible graphs, whose embedding functor is denoted by $\imath_0$. According to item (2) in Proposition
\ref{prp:graph_decloning}, a functor $\rho:\Gr\to\Grz$ is defined for objects~$\Gamma$ of $\Gr$ by
$\rho(\Gamma):=\ul\Gamma$ and for morphisms $\Phi:\Gamma\to\Gamma'$ by
$\rho(\Phi):=\ul\Phi:\ul\Gamma\to\ul{\Gamma'}$; we call $\rho$ the \emph{graph~decloning~functor}. In categorical
language, item (3) in Proposition \ref{prp:graph_decloning} implies, upon taking
$\wght_\Gamma=\wght_{\Gamma'}=\wght_0$ that $\rho$ is a full functor. Also item (1) in Proposition
\ref{prp:graph_decloning} says that $\Grz$ is a reflective subcategory of~$\Gr$, with reflection functor the graph
decloning functor $\rho$. Said differently, $\rho$ is an adjoint functor for~$\imath_0$. Among other properties, it
shows that decloning is a natural operation on graphs; in the next section we will see that decloning of LV systems
has similar functorial properties.
\section{Lotka-Volterra systems}\label{sec:LV}
\subsection{Basic definitions and properties}\label{par:basic}
We first recall the basic definitions and properties of skew-symmetric Lotka-Volterra systems. Let $A=(a_{i,j})$ be
any skew-symmetric $n\times n$ matrix with entries in $\bbF$, where $n$ is any positive integer. The Lotka-Volterra
system defined by $A$ is the Hamiltonian system on $\bbF^n$, defined by the following quadratic vector field:
\begin{equation}\label{eq:LV_System}
  \dot{x}_{i} = \sum_{j=1}^{n} a_{i,j}\,x_{i}x_{j}\;, \qquad \text{for}\; i=1, \ldots, n\;.
\end{equation}
We have denoted the standard coordinates on $\bbF^n$ by $x_1,\dots,x_n$. The matrix~$A$ also defines a Poisson
structure on $\bbF^n$,
\begin{equation}\label{eq:LV_pi}
  \pi_A:=\sum_{i<j}a_{i,j}x_ix_j\pp{}{x_i}\we\pp{}{x_j}\;.
\end{equation}%
The corresponding Poisson bracket is given by $\pb{x_i,x_j}_A=a_{i,j}x_ix_j$. It is a quadratic Poisson bracket,
known as a \emph{diagonal} Poisson bracket because of its particular form; the Poisson structure $\pi_A$ is also
said to be \emph{log-canonical} because it satisfies $\pb{\log x_i,\log x_j}_A=a_{i,j}$. If we denote by $H$ the
sum of the coordinates, $H:=x_1+x_2+\dots+x_n$, then it is clear from (\ref{eq:LV_pi}) that~(\ref{eq:LV_System}) is
the Hamiltonian vector field $\X_H=\Pb H_A$. We call the Hamiltonian system $(\bbF^n,\pi_A,H)$ a
\emph{(skew-symmetric) Lotka-Volterra system}, or simply an \emph{LV system}; notice that it is entirely determined
by $A$.

Several basic properties of the LV system associated with $A$ can be read off from the matrix $A$. The
rank of $\pi_A$ (which is by definition the maximal rank of $\pi_A$ at the points of~$\bbF^n$) is equal to the rank
of $A$. Moreover, if $(\a_1,\dots,\a_n)\in\bbN^n$ is a null-vector of~$A$, then $C:=x_1^{\a_1}x_2^{\a_2}\dots
x_n^{\a_n}$ is a Casimir function of $\pi_A$, i.e., $\X_C=0$. This is also true if $(\a_1,\dots,\a_n)\in\bbF^n$
upon properly interpreting $C$ as a function on an open subset of $\bbF^n$. When all the entries of $A$ are
integers or rational numbers, as will be the case in the examples which we will discuss in
Section \ref{sec:int_lax}, a basis for the null-vectors of~$A$ can be chosen in $\bbN^n$,
yielding $n-\Rk A$ independent rational Casimir functions of~$\pi_A$.
%

The matrix $A$ is also useful for describing the reductions which are obtained by setting one or several of the
coordinates $x_i$ equal to zero. Notice that such a reduction is a special type of \emph{Poisson reduction} since
for any collection $I\subset\set{1,2,\dots,n}$ the subspace of $\bbF^n$, defined by $x_i=0$ for all $i\in I$, is a
Poisson submanifold of $(\bbF^n,\pi_A)$. The reduced Poisson structure is therefore just the restriction of $\pi_A$
to the subspace. In particular, the reduced system is also an LV system and its matrix is obtained by removing from
$A$ the $i$-th row and $i$-th column, for all $i\in I$, with as Hamiltonian the sum of all remaining coordinate
functions.  In general, the rank of the Poisson structure drops along the subspace, as it is equal to the rank of
its defining matrix.

A \emph{smooth morphism of LV systems} $\phi:(\bbF^n,\pi_A,H)\to(\bbF^{n'},\pi_{A'},H')$ is a smooth\footnote{By
a slight abuse of terminology, we use the terms \emph{smooth} and \emph{diffeomorphic} both in the case of
$\bbF=\bbR$ and $\bbF=\bbC$; strictly speaking we should say in the latter case \emph{holomorphic} and
\emph{biholomorphic}.}  map $\phi:\bbF^n\to\bbF^{n'}$, which preserves the Poisson structure and the Hamiltonian;
in terms of formulas this means that $\pb{\phi^*F,\phi^*G}_A=\phi^*\pb{F,G}_{A'}$ for any smooth functions $F$ and
$G$ on $\bbF^{n'}$, and that $\phi^*H'=H$. It leads to the notion of \emph{smooth isomorphism} of LV systems.

We will in what follows only consider linear morphisms and isomorphisms of LV systems. The reason for
this is given by the following proposition which says that if two LV systems are smoothly isomorphic,
they are linearly isomorphic.
\begin{prop}\label{prp:diffeo_to_linear}
  Suppose that $\phi:(\bbF^n,\pi_A,H)\to(\bbF^{n},\pi_{A'},H)$ is a smooth isomorphism of LV
  systems. Then $\diff_0\phi$, the differential of $\phi$ at the origin, is a (linear) isomorphism of
  LV systems.
\end{prop}
\begin{proof}
We denote the standard coordinates on $\bbF^n$ by $x_1,x_2,\dots,x_n$. We use the fact that any smooth function
$F:\bbF^n\to\bbF$ can be written as $F=Q+L+c$, where $c:=F(0)\in\bbF$ and $L:=\sum_{k=1}^nx_k\pp F{x_k}(0)$, so that
$Q$ is a smooth function with $Q(0)=0$ and $\pp Q{x_k}(0)=0$ for $k=1,\dots,n$. Applied to the functions
$\phi^*x_i:\bbF^n\to\bbF$ we write $\phi^*x_i=Q_i+L_i+c_i$, for $i=1,\dots,n$. Then the linear map
$\diff_0\phi:\bbF^n\to\bbF^n$ is given by $(L_1,\dots,L_n)$. It is an isomorphism because $\phi$ is a
diffeomorphism. We need to show that
\begin{equation}\label{eq:lin_poi}
  \pb{L_i,L_j}_A=a'_{i,j}L_iL_j\;, \qquad \text{for}\; i,j=1, \ldots, n\;.
\end{equation}
When one of $x_i$ and $x_j$, say $x_i$, is a Casimir function of $\pi_{A'}$, then $\phi^*x_i$ is a Casimir function
of $\pi_A$, hence also its linear part $L_i$ (for degree reasons); also $a'_{i,j}=0$ for all $j$, so both sides
vanish. It remains to prove (\ref{eq:lin_poi}) when $x_i$ and $x_j$ are not Casimir functions of $\pi_{A'}$. We
first show that if $x_i$ is not a Casimir function, then $c_i=0$. The proof goes by contradiction. Suppose that
$x_i$ is not a Casimir function of $\pi_{A'}$ and that $c_i\neq0$. Then there exists a $k$ with $1\leqs k\leqs n$
such that $a'_{i,k}\neq0$.  On the one hand,
\begin{equation*}
  \phi^*\pb{x_i,x_k}_{A'}=a'_{i,k}\(\phi^*x_i\)\(\phi^*x_k\)=a'_{i,k}(Q_i+L_i+c_i)(Q_k+L_k+c_k)\;,
\end{equation*}%
while on the other hand
\begin{equation*}
  \pb{\phi^*x_i,\phi^*x_k}_A=\pb{Q_i,Q_k}_A+\pb{Q_i,L_k}_A+\pb{L_i,Q_k}_A+\pb{L_i,L_k}_A\;.
\end{equation*}%
Since $\phi$ is a Poisson map, the above right hand sides are equal. Since the latter has no constant or linear
terms, $a'_{i,k}c_ic_k=0$ and $a'_{i,k}(c_kL_i+c_iL_k)=0$, so that $c_k=0$ and $L_k=0$. The latter is impossible,
since $\diff_0\phi=(L_1,\dots,L_n)$ is an isomorphism. We have arrived at a contradiction and may conclude that if
$x_i$ is not a Casimir function, then $c_i=0$. Let us suppose now that $x_i$ and $x_j$ are not Casimir functions,
so that $c_i=c_j=0$. Then the quadratic parts of $\phi^*\pb{x_i,x_j}_{A'}$ and $\pb{\phi^*x_i,\phi^*x_j}_A$ are
respectively given by $a'_{i,j}L_iL_j$ and $\pb{L_i,L_j}_{A}$; they are equal since $\phi$ is a Poisson map, so
that $\pb{L_i,L_j}_A=a'_{i,j}L_iL_j$. This proves~(\ref{eq:lin_poi}), hence that $\diff_0\phi$ is a (linear)
Poisson diffeomorphism.
\end{proof}

\subsection{Lotka-Volterra systems associated with graphs}\label{par:graphs}
Since skew-sym\-metric graphs and LV systems both bijectively correspond to skew-symmetric matrices, it
is natural to think of LV systems as being associated with graphs. Since the rows and columns of the
adjacency matrix of a skew-symmetric graph $\Gamma=(S,A)$ are indexed by the vertex set~$S$ of $\Gamma$ we take a
more intrinsic point of view and define the \emph{LV system}
associated with $\Gamma$, to be the Hamiltonian system $(\bbF[S],\pi_A,H_S)$; it is denoted by $\LV(\Gamma)$. Here,
$\bbF[S]$ stands for the vector space generated by $S$, with (unordered) basis $S$,
\begin{equation*}
  \bbF[S]=\set{\sum_{s\in S}\a_ss\mid \a_s\in \bbF\hbox{ for all } s\in S}\;.
\end{equation*}%
Elements of $\bbF[S]$ are also denoted in (unordered) vector form $(\a_s)_{s\in S}$. The linear coordinates on
$\bbF[S]$ associated with $S$ are denoted $x_s,\ s\in S$. Also, $\pi_A$ is the Poisson structure on $\bbF[S]$ whose
associated Poisson bracket is given~by
\begin{equation}\label{eq:PB_S}
  \pb{x_s,x_t}_A=a_{s,t}x_sx_t\;, \qquad \text{for}\; s,t\in S\;.
\end{equation}
%
The Hamiltonian $H_S$ is the sum of the coordinates, $H_S:=\sum_{s\in S}x_s$.  In this notation, the LV vector
field $\X_{H_S}$ takes the form

\begin{equation}\label{eq:LV_S}
  \dot{x}_{s} = x_{s}\sum_{t\in S}a_{s,t}\,x_{t}\;, \qquad \text{for}\; s\in S\;.
\end{equation}

When $\Gamma$ is equipped with a weight vector $\wght$, we will say that
$(\LV(\Gamma),\wght)$ is a \emph{weighted LV system}. Notice that replacing $S$ by $\set{1,2,\dots,n}$ amounts to
totally ordering the elements of $S$, which can be done in $\#S\;!$ ways, making the link between skew-symmetric
graphs and LV systems less intrinsic.

A \emph{morphism of Lotka-Volterra systems}, or simply an \emph{LV morphism},
$\phi:(\bbF[S],\pi_A,H_S)\to(\bbF[S'],\pi_{A'},H_{S'})$ is a linear Poisson map
$\phi:(\bbF[S],\pi_A)\to(\bbF[S'],\pi_{A'})$ for which $\phi^*H_{S'}=H_S$. Notice that $\phi$ does not entail a map
from $S$ to $S'$, so there is no natural notion of a morphism of weighted LV systems; we will come back to this in
Section \ref{par:LV_isomorphisms}.

Several properties of a given LV system are more easily seen or explained from the underlying graph
$\Gamma$ than from its adjacency matrix~$A$. A simple example of this is that if $\Gamma$ is not connected,
$\LV(\Gamma)$ is a direct product (as a Hamiltonian system) of the LV systems which correspond to the
connected components of the graph; also, reduction to the subspace defined by $x_s=0$ for $s\in S_0\subset S$
corresponds to removing from $\Gamma$ all vertices labeled by the entries of $S_0$ (one removes of course also all
arrows which are incident with these vertices). Another example spelled out in Proposition~\ref{prp:functor} below
is that graph morphisms lead to morphisms of LV systems.  Also, the cloning and decloning constructions
for LV systems which are introduced in the next subsection derive naturally from the corresponding
constructions for skew-symmetric graphs.

\vfill

A few examples of graphs corresponding to some well-known examples are given in the three figures which follow (see
Section \ref{par:examples} for more information on these examples).


%
%

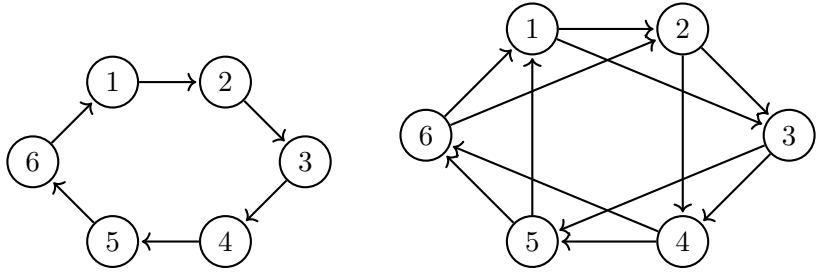
\begin{figure}[h]
  \begin{tikzpicture}[->,shorten >=1pt,auto,node distance=1.5cm,
                thick,main node/.style={circle,draw,font=\bfseries}]
  \node[main node] (1) {$1$};
  \node[main node] (2) [right of=1] {$2$};
  \node[main node] (3) [below right of=2] {$3$};
  \node[main node] (4) [below left of=3] {$4$};
  \node[main node] (5) [left of=4] {$5$};
  \node[main node] (6) [below left of=1] {$6$};
  \path
    (1) edge (2)
    (2) edge (3)
    (3) edge (4)
    (4) edge (5)
    (5) edge (6)
    (6) edge (1);
  \end{tikzpicture}
  \qquad
  \begin{tikzpicture}[->,shorten >=1pt,auto,node distance=2cm,
                thick,main node/.style={circle,draw,font=\bfseries}]
  \node[main node] (1) {$1$};
  \node[main node] (2) [right of=1] {$2$};
  \node[main node] (3) [below right of=2] {$3$};
  \node[main node] (4) [below left of=3] {$4$};
  \node[main node] (5) [left of=4] {$5$};
  \node[main node] (6) [below left of=1] {$6$};
  \path
    (1) edge (2) edge (3)
    (2) edge (3) edge (4)
    (3) edge (4) edge (5)
    (4) edge (5) edge (6)
    (5) edge (6) edge (1)
    (6) edge (1) edge (2);
  \end{tikzpicture}
  \caption{The LV system corresponding to the left graph is the Kac-van Moerbeke system $\KM(6)$; the
    one corresponding to the right graph is $\B62$.\label{fig:1st}}
\end{figure}

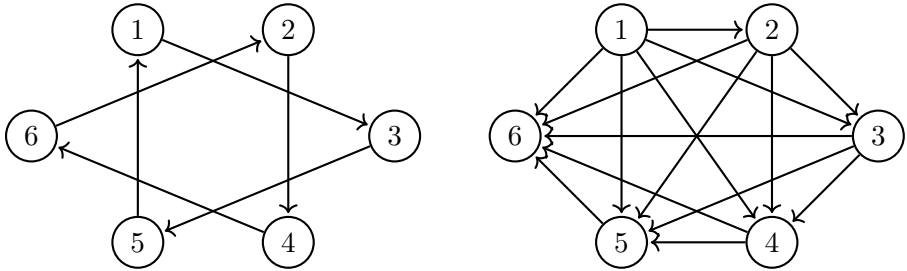
\begin{figure}[h]
  \begin{tikzpicture}[->,shorten >=1pt,auto,node distance=2cm, thick,main node/.style={circle,draw,font=\bfseries}]
  \node[main node] (1) {$1$};
  \node[main node] (2) [right of=1] {$2$};
  \node[main node] (3) [below right of=2] {$3$};
  \node[main node] (4) [below left of=3] {$4$};
  \node[main node] (5) [left of=4] {$5$};
  \node[main node] (6) [below left of=1] {$6$};
  \path
    (1) edge (3)
    (2) edge (4)
    (3) edge (5)
    (4) edge (6)
    (5) edge (1)
    (6) edge (2);
\end{tikzpicture}
\qquad
\begin{tikzpicture}[->,shorten >=1pt,auto,node distance=2cm,
                thick,main node/.style={circle,draw,font=\bfseries}]
  \node[main node] (1) {$1$};
  \node[main node] (2) [right of=1] {$2$};
  \node[main node] (3) [below right of=2] {$3$};
  \node[main node] (4) [below left of=3] {$4$};
  \node[main node] (5) [left of=4] {$5$};
  \node[main node] (6) [below left of=1] {$6$};
  \path
    (1) edge (2) edge(3) edge (4) edge (5) edge (6)
    (2) edge (3) edge (4) edge (5) edge (6)
    (3) edge (4) edge (5) edge (6)
    (4) edge (5) edge (6)
    (5) edge (6);
\end{tikzpicture}
\caption{The LV system corresponding to the left graph is the direct sum of two copies of $\KM(3)$; the
  one on the right corresponds to $\LV(6,0)$.\label{fig:2nd}}
\end{figure}

\begin{figure}[h]
  \begin{tikzpicture}[->,shorten >=1pt,auto,node distance=2cm,
                thick,main node/.style={circle,draw,font=\bfseries}]
  \node[main node] (1) {$1$};
  \node[main node] (2) [right of=1] {$2$};
  \node[main node] (3) [below right of=2] {$3$};
  \node[main node] (4) [below left of=3] {$4$};
  \node[main node] (5) [left of=4] {$5$};
  \path
    (1) edge (2) edge(3)
    (2) edge (3) edge (4)
    (3) edge (4) edge (5)
    (4) edge (5)
    (5) edge (1) ;
  \end{tikzpicture}
\qquad
\begin{tikzpicture}[->,shorten >=1pt,auto,node distance=2cm,
                thick,main node/.style={circle,draw,font=\bfseries}]
  \node[main node] (1) {$1$};
  \node[main node] (2) [right of=1] {$2$};
  \node[main node] (3) [below right of=2] {$3$};
  \node[main node] (4) [below left of=3] {$4$};
  \node[main node] (5) [left of=4] {$5$};
  \path
    (1) edge (2) edge(3) edge (4) edge (5)
    (2) edge (3) edge (4) edge (5)
    (3) edge (4) edge (5)
    (4) edge (5) ;
\end{tikzpicture}
\caption{The two graphs are obtained by removing the vertex $6$ from the right graph in Figures \ref{fig:1st} and
  \ref{fig:2nd}, respectively. The corresponding LV systems are reductions of the original ones. Notice
  that the right one is $\LV(5,0)$.}\label{fig:3rd}
\end{figure}
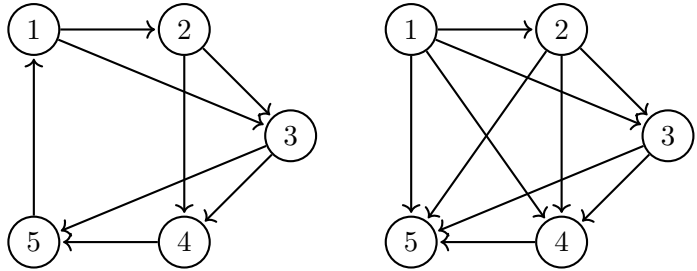


We show in the following proposition that graph morphisms lead to LV morphisms between the corresponding
LV systems. Any map $\Phi:S\to S'$ leads to a unique linear map $\phi:\bbF[S]\to\bbF[S']$, extending
$\Phi$, i.e., $\phi(s)=\Phi(s)$ for all $s\in S$. It is called the \emph{linear extension} of $\Phi$. When
$\Phi:\Gamma\to\Gamma'$ is a graph morphism, we denote its linear extension by $\LV(\Phi)$.
\begin{prop}\label{prp:functor}
  Let $\Gamma$ and $\Gamma'$ be two skew-symmetric graphs and suppose that $\Phi:\Gamma\to\Gamma'$ is
  a graph morphism.
  \begin{enumerate}
    \item $\LV(\Phi):\LV(\Gamma)\to\LV(\Gamma')$ is an LV morphism.
    \item If $\Phi':\Gamma'\to\Gamma''$ is another graph morphism, then
        $\LV(\Phi'\circ\Phi)=\LV(\Phi')\circ\LV(\Phi)$; also, $\LV(\Id_\Gamma)=\Id_{\LV(\Gamma)}$.
    \item $\Phi$ is a graph isomorphism if and only if $\LV(\Phi)$ is an LV isomorphism.
  \end{enumerate}
\end{prop}
\begin{proof}
Let us write $\Gamma=(S,A)$ and $\Gamma'=(S',A')$, and let us denote the linear map $\LV(\Phi)$ by $\phi$. The
linear coordinates on~$\bbF[S]$ and on $\bbF[S']$ are respectively denoted by $x_s,$ with $s\in S$ and $y_u,$ with
$u\in S'$. By definition, $\phi$ is given by $\phi\(\sum_{s\in S}\a_ss\)=\sum_{s\in S}\a_s\Phi(s)$, where
$\a_s\in\bbF$ for all $s\in S$; in vector form, this is written
$\phi(\a_s)_{s\in{S}}=(\sum_{s\in{S}}\a_s\delta_{u,\Phi(s)})_{u\in S'}$. We show that the corresponding algebra
homomorphism $\phi^*$ is given, for $u\in S'$, by
\begin{equation}\label{eq:lin_functions}
  \phi^*y_u=\sum_{\Phi(s)=u}x_s\;.
\end{equation}%
To prove this formula, which is an equality of linear functions on $\bbF[S]$, it suffices to check that both sides
of (\ref{eq:lin_functions}) take the same value when evaluated on any element of $S$. For $t\in S$, we have
\begin{align*}
  &(\phi^*y_u)(t)=y_u(\phi(t))=y_u(\Phi(t))=\delta_{u,\Phi(t)}\;,\\
  &\(\sum_{\Phi(s)=u}x_s\)(t)=\sum_{\Phi(s)=u} x_s(t)=\sum_{\Phi(s)=u}\delta_{s,t}=\delta_{u,\Phi(t)}\;,
\end{align*}
which proves \eqref{eq:lin_functions}. To show that $\phi$ is a Poisson map, we need to verify that
\begin{equation}\label{eq:poisson_morphism}
  \pb{\phi^*y_u,\phi^*y_v}_{A}=\phi^*\pb{y_u,y_v}_{A'}\;
\end{equation}%
for all $u,v\in S'$. First, notice that since $\pb{y_u,y_v}_{A'}$ is a multiple of $y_uy_v$, both sides of
(\ref{eq:poisson_morphism}) are zero when $u$ or $v$ do not belong to $\Im(\Phi)$. We may therefore suppose that
$u,v\in\Im(\Phi)$. Since $\Phi$ is a graph morphism, if $\Phi(s)=u$ and $\Phi(t)=v$, then
$a'_{u,v}=a_{s,t}$. Therefore,
\begin{align*}
  \pb{\phi^*y_u,\phi^*y_v}_A
  &=\sum_{\Phi(s)=u}\sum_{\Phi(t)=v}\pb{x_s,x_t}_A=\sum_{\Phi(s)=u}\sum_{\Phi(t)=v}a_{s,t}x_sx_t\\
  &=a'_{u,v}\(\sum_{\Phi(s)=u}x_s\)\(\sum_{\Phi(t)=v}x_t\)=a'_{u,v}\phi^*(y_u)\phi^*(y_v)\\
  &=\phi^*(a'_{u,v}y_uy_v)=\phi^*\pb{y_u,y_v}_{A'}\;.
\end{align*}%
Moreover,
\begin{equation*}
  \phi^*H_{S'}=\sum_{u\in S'}\phi^*y_u=\sum_{u\in S'}\sum_{\Phi(s)=u}x_s=\sum_{s\in S}x_s=H_S\;.
\end{equation*}%
It follows that $\phi$ is an LV morphism, which is item (1) (recall that $\phi=\LV(\Phi)$). For the proof of item
(2), it suffices to verify that $\LV(\Phi'\circ\Phi)$ and $\LV(\Phi')\circ\LV(\Phi)$ take the same value at every
$s\in S\subset\bbF[S]$, just like $\LV(\Id_\Gamma)$ and $\Id_{\LV(\Gamma)}$, but that is trivial. The proof of the
direct implication in item (3) follows from items (1) and~(2); the inverse implication in item (3) follows from the
fact that if $\LV(\Phi)$ is an isomorphism, then $\Phi$ is bijective since $\LV(\Phi)$ is the linear extension
of~$\Phi$.
\end{proof}
Items (1) and (2) of the proposition imply that $\LV$ is a functor. We will come back to this in Section
\ref{par:LV_cat}.

Not all morphisms between LV systems are induced by graph morphisms, even when the underlying graphs
are irreducible. A simple counterexample is given in Figure \ref{fig:cex}.

\begin{figure}[h]
  \begin{tikzpicture}[->,shorten >=1pt,auto,node distance=1.5cm, thick,main node/.style={circle,draw,font=\bfseries}]
  \node[main node] (1) {$s$};
  \node[main node] (2) [below left of=1] {$t$};
  \node[main node] (3) [below right of=1] {$u$};
  \node (4) [right of=3] {};
  \node[main node] (5) [right of=4] {$v$};
  \node[main node] (6) [right of=5] {$w$};
  \path
  (1) edge (2)
  (2) edge (3)
  (1) edge (3)
  (5) edge (6);
  \end{tikzpicture}
  \caption{There exists no graph morphism between the left and right graphs. However, a morphism $\phi$ between
    their LV systems is defined by $\phi^*y_v=x_s+x_t$ and \hbox{$\phi^*y_w=x_u$}.\label{fig:cex}}
\end{figure}
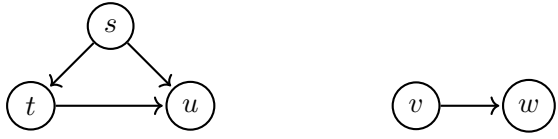

\subsection{Cloning and decloning of Lotka-Volterra systems}\label{par:LV_cloning}
We define in this paragraph the cloning and decloning of LV systems, leaving the more delicate issue of the
decloning of LV morphisms to Section \ref{par:LV_morphisms_decloning}. We do this by using the corresponding
constructions for the underlying graphs. Let $\Gamma$ be a skew-symmetric graph and let $\wght$ be a weight vector
for $\Gamma$. Recall that we denote by $\Gamma^\wght$ the cloning of $(\Gamma,\wght)$ and by $\ul\Gamma$ and
$(\ul\Gamma,\wght_\Gamma)$ respectively the decloning and the weighted decloning of $\Gamma$. Their LV systems are
by definition the \emph{cloning} of the LV system $\LV(\Gamma)$ with weight vector $\wght$, respectively the
\emph{(weighted) decloning} of $\LV(\Gamma)$. Thus, the cloning of the weighted LV system $(\LV(\Gamma),\wght)$ is
$\LV(\Gamma^\wght)$ and the decloning, respectively weighted decloning, of $\LV(\Gamma)$ is $\LV(\ul\Gamma)$,
respectively $(\LV(\ul\Gamma),\wght_\Gamma)$. When $\Gamma$ is irreducible, $\Gamma\simeq\ul\Gamma$, we will say
that $\LV(\Gamma)$ is \emph{irreducible}, otherwise that it is \emph{reducible}.

We give a more explicit description of these systems, which we do first in the case of cloning. Let
$(\Gamma,\wght)$ be a weighted graph, with $\Gamma=(S,A)$. Recall that $\Gamma^\wght$ is the skew-symmetric graph
$(S^\wght,A^\wght)$, where the elements of $S^\wght$ are all $s_i$ with $s\in S$ and $1\leqs i\leqs\wght(s)$. Thus,
the phase space $\bbF[S^\wght]$ of $\LV(\Gamma^\wght)$ has dimension $\vert\wght\vert$ and is equipped with
coordinate functions $x_{s_i}$, with $s$ and $i$ as above. Since the entries of $A^\wght$ are given by
$a^\wght_{s_i,t_j}:=a_{s,t}$, for $s,t\in S$ and $1\leqs i\leqs\wght(s)$, $1\leqs j\leqs\wght(t)$, the Poisson
structure~$\pi_{A^\wght}$ is given, as in \eqref{eq:PB_S}, by the following diagonal Poisson brackets:
\begin{equation}\label{eq:pb_wght}
  \pb{x_{s_i},x_{t_j}}_{A^\wght}=a_{s,t}x_{s_i}x_{t_j}\;.
\end{equation}%
In this formula, $s,t\in S$ and $1\leqslant i\leqslant\wght(s)$ and $1\leqslant j\leqslant\wght(t)$. Since the
matrices~$A$ and $A^\wght$ have the same rank, the corresponding Poisson structures $\pi_A$ and~$\pi_{A^\wght}$
have the same rank (equal to the rank of $A$). Since for any $s\in S$, the lines with labels $s_1,s_2,\dots
s_{\wght(s)}$ are the same, the functions $x_{s_i}/x_{s_1}$ are (independent) Casimir functions of $\pi_{A^\wght}$,
for $s\in S$ and $i=2,3,\dots,\wght(s)$; this follows also easily from (\ref{eq:pb_wght}). Since the Hamiltonian of
$\LV(\Gamma^\wght)$ is the sum of the coordinate functions, it is given by
\begin{equation}\label{eq:H_wght}
  H_{S^\wght}=\sum_{s\in S}\sum_{i=1}^{\wght(s)} x_{s_i}\;,
\end{equation}%
and the Hamiltonian vector field of ${H_{S^\wght}}$ takes the simple form
\begin{equation}\label{eq:LV_wght}
  \dot{x}_{s_i} = x_{s_i}\sum_{t\in S}\sum_{j=1}^{\wght(t)}a_{s,t}\,x_{t_j}\;, \qquad \text{for}\; s\in S
  \text{ and } i=1,\dots,\wght(s)\;.
\end{equation}

We now give a more explicit description of decloning of LV systems, which we do in two different
ways. We first show that decloning amounts to a special type of Poisson reduction \cite[Section 5.2.2]{PLV}, given
by a Poisson submersion, as stated in the following proposition:
\begin{prop}\label{prp:decloning}
  Let $\Gamma$ be any graph and let $p:\Gamma\to\ul\Gamma$ denote the decloning map of $\Gamma$. Then
  $\LV(p):\LV(\Gamma)\to\LV(\ul\Gamma)$ is an LV morphism, which is a surjective submersion, so that
  $\LV(\ul\Gamma)$ is obtained from $\LV(\Gamma)$ by Poisson reduction. Explicitly, $\LV(p)$ is given by
  \begin{equation}\label{eq:LV(p)}
    \begin{array}{lcccl}
      \LV(p)&:&\bbF[S]&\to&\bbF[\ul S]\\
      & &(\a_s)_{s\in S} &\mapsto&
         \left(\sum_{s\in\ul s}\a_s\right)_{\ul s\in \ul S}\;,
    \end{array}
  \end{equation}
  and is called the \emph{decloning map of} $\LV(\Gamma)$.
\end{prop}
\begin{proof}
Since $p:\Gamma\to\ul\Gamma$ is a graph morphism, $\LV(p)$ is an LV morphism (Proposition \ref{prp:functor}). Since
$p$ is surjective, $\LV(p)$ is surjective on $\ul S$, so it is a surjective linear map, hence a (surjective)
submersion. The explicit formula for $\LV(p)$ is an easy transcription of (\ref{eq:lin_functions}).
\end{proof}
Alternatively, one can describe decloning of LV systems as a different type of Poisson reduction, obtained by
putting several of the coordinates $x_{s_i}$ equal to zero; recall from Section \ref{par:basic} that this amounts
to restricting the LV system to an LV system on a subspace, and recall from Section \ref{par:graphs} that this
amounts to removing from the graph $\Gamma$ the corresponding vertices~$s_i$. Applied to the current case, in which
$\ul\Gamma$ is obviously isomorphic to the subgraph $\Gamma_0$ of~$\Gamma$, induced by any set $S_0$ of
representatives of $S$ modulo~$\sim$, we get that $\LV(\ul\Gamma)\simeq\LV(\Gamma_0)$ and that $\LV(\Gamma_0)$ is
obtained by Poisson reduction from~$\LV(\Gamma)$.

The above definitions of cloning, decloning and irreducibility of LV systems are a priori unsatisfactory because we
have not shown yet that if $\LV(\Gamma)\simeq\LV(\Gamma')$ then $\Gamma\simeq\Gamma'$, and hence that
$\LV(\ul\Gamma)\simeq\LV(\ul{\Gamma'}).$ To do this, one needs to study LV morphisms in more detail, which will be
done in the next subsection.

\subsection{Decloning of morphisms of Lotka-Volterra systems}\label{par:LV_morphisms_decloning}

We show in this subsection that surjective morphisms of LV systems can be decloned. To do this, and in order to
prove some related results, we will use the following key lemma:
\begin{lemma}\label{lma:pppambos}
  Let $\phi:(\bbF[S],\pi_A)\to (\bbF[S'],\pi_{A'})$ be a linear Poisson morphism. It is assumed that
  $\Im\phi$ is not contained in a coordinate hyperplane $y_u=0$, $u\in S'$. We denote, as before, the
  coordinates on $\bbF[S]$ and on $\bbF[S']$ by~$x_s$, with $s\in S$, respectively by $y_u$, with $u\in S'$. Let
  $B=(\b_{u,s})$ be the matrix defined by $\phi^*y_u=\sum_{s\in S}\b_{u,s}x_s$. Suppose that
  $\b_{u,s}\b_{v,s}\neq0$ for some \hbox{$u\neq v\in S'$} and some $s\in S$. Then $u\sim v$.
\end{lemma}
\begin{proof}
We first express the fact that $\phi$ is a Poisson map. For $u,v\in S'$ we have that
\begin{align*}
  \phi^*\pb{y_u,y_v}_{A'}&=a'_{u,v}(\phi^*y_u)(\phi^*y_v)=\sum_{s,t\in S}a'_{u,v}\b_{u,s}\b_{v,t}x_sx_t\;,\\
  \pb{\phi^*y_u,\phi^*y_v}_A&=\sum_{s,t\in S}a_{s,t}\b_{u,s}\b_{v,t}x_sx_t\;.
\end{align*}%
Taking the coefficient of $x_sx_t$ in these expressions we find that $\phi$ is a Poisson map if and only if
\begin{equation}\label{eq:poisson_map}
  (a'_{u,v}-a_{s,t})\b_{u,s}\b_{v,t}+  (a'_{u,v}+a_{s,t})\b_{u,t}\b_{v,s}=0\;,\quad \text{for all}\;s,t\in
  S,\;u,v\in S'\;.
\end{equation}%
Suppose now that $\phi$ is a Poisson map and that $\b_{u,s}\b_{v,s}\neq0$ for some $u\neq v\in S'$ and some
$s\in S$. We show that $u\sim v$, i.e., that $a'_{u,w}=a'_{v,w}$ for all $w\in S'$. Replacing $t$ by $s$ in
(\ref{eq:poisson_map}) we find, since $a_{s,s}=0$, that $2a'_{u,v}\b_{u,s}\b_{v,s}=0$, so that
$a'_{u,v}=0=a'_{v,v}$, which shows that $a'_{u,w}=a'_{v,w}$ for $w=v$, and also for $w=u$. Let $w\in S'$, different
from $u$ and $v$. We distinguish two cases. Suppose first that $\b_{w,s}=0$ and let $t$ be such that
$\b_{w,t}\neq0$; such a $t$ exists because $\Im\phi$ is not contained in the coordinate hyperplane $y_w=0$. If
we replace $v$ by $w$ in (\ref{eq:poisson_map}) we get $(a'_{u,w}-a_{s,t})\b_{u,s}\b_{w,t}=0$, so that
$a'_{u,w}=a_{s,t}$. Similarly, if we replace $u$ by $w$ in (\ref{eq:poisson_map}) we get $a'_{v,w}=a_{s,t}$. It
follows that $a'_{u,w}=a'_{v,w}$. When $\b_{w,s}\neq0$ we take $t=s$ and proceed as in the first case,
with $v=w$ (resp. $u=w$), to find $a'_{u,w}=0$ (resp. $a'_{v,w}=0$).
This leads again to $a'_{u,w}=a'_{v,w}$ and completes the proof that $u\sim v$.
\end{proof}
Using the lemma, we find a simple description of surjective LV morphisms, when the target system is irreducible.
\begin{prop}\label{prp:normal_form}
  Let $\phi:(\bbF[S],\pi_A,H_S)\to (\bbF[S'],\pi_{A'},H_{S'})$ be a surjective LV morphism, where the target system
  is supposed irreducible.
  \begin{enumerate}
    \item
      There exists for any $u\in S'$ a (unique) non-empty subset $S_u$ of $S$, such that $\phi:\bbF[S]\to\bbF[S']$
      is given by
      \begin{equation}\label{eq:normal_form}
        \phi\(\a_s\)_{s\in S}=\(\sum_{s\in S_u}\a_s\)_{u\in S'}\;;
      \end{equation}%
      the subsets $(S_u)_{u\in S'}$ form a partition of $S$; in particular, if $u\neq v\in S'$ then $S_u\cap
      S_v=\emptyset$.
    \item For any $u\neq v\in S'$, one has $a_{s,t}=a'_{u,v}$ for all $s\in S_u$ and $t\in S_v$; in particular,
      $a_{s,t}$ is independent of $s\in S_u$ and of $t\in S_v$.
  \end{enumerate}
\end{prop}
\begin{proof}
In order to prove item (1), we apply Lemma \ref{lma:pppambos} in case the target system is irreducible, i.e., the
graph $(S',A')$ is irreducible. Then $u\sim v$ for $u,v\in S'$ implies that $u=v$, so that by the lemma,
$\b_{u,s}\b_{v,s}=0$ for all $s\in S$ and all $u\neq{}v\in{S'}$. If we define $S_u:=\set{s\in S\mid \b_{u,s}\neq0}$
for $u\in S'$ then this implies that if $u\neq v\in S'$ then $S_u\cap S_v=\emptyset$. It follows from the
definition of $S_u$ that $\phi^*y_u=\sum_{s\in S_u}\b_{u,s}x_s,$ where the scalars $\b_{u,s}$, with $s\in S_u$, are
non-zero. Since every $s\in S$ appears in at most one $S_u$, we can deduce from
\begin{equation*}
  \sum_{s\in S}x_s=H_S=\phi^*H_{S'}=\sum_{u\in S'}\phi^*y_u=\sum_{u\in S'}\sum_{s\in S_u}\b_{u,s}x_s\;
\end{equation*}%
that all $\b_{u,s}$ with $s\in S_u$ are equal to $1$, and that every element $s$ of $S$ belongs to one of the
$S_u$. It follows that
\begin{equation}\label{eq:phi_star}
  \phi^*y_u=\sum_{s\in S_u}x_s\;,
\end{equation}
which is equivalent to (\ref{eq:normal_form}). Moreover, the surjectivity of $\phi$ implies that each subset $S_u$
is non-empty, so that $(S_u)_{u\in{}S'}$ is a partition of $S$. This proves item (1).

In order to prove item (2), let $u\neq v\in S'$. Then, in view of (\ref{eq:phi_star}), and since $\phi$ is a
Poisson map,
\begin{equation*}
  \sum_{s\in S_u}\sum_{t\in S_v}a_{s,t}x_sx_t=\pb{\phi^*y_u,\phi^*y_v}_{A}
  =\phi^*\pb{y_u,y_v}_{A'}=\sum_{s\in S_u}\sum_{t\in S_v}a'_{u,v}x_sx_t\;.
\end{equation*}%
Since $S_u\cap S_v=\emptyset$, this implies item (2).
\end{proof}
Before proving that surjective LV morphisms can be decloned, we prove a lemma which is an analog of the first item
of Lemma \ref{lma:graph_morphisms}.

\begin{lemma}\label{lma:LV_morphisms}
  Let $\Gamma$ and $\Gamma'$ be skew-symmetric graphs, with decloning maps $p:\Gamma\to\ul\Gamma$
  and $p':\Gamma'\to\ul{\Gamma'}$. Suppose that $\phi$ and $\psi$ are two linear maps, making the
  following diagram commutative:
  \begin{equation}\label{dia:LV_lma}
    \begin{tikzcd}[row sep=5ex, column sep=7ex]
      \LV(\Gamma)\arrow{r} {\phi}\arrow[swap]{d} {\LV(p)}&\LV(\Gamma')\arrow{d}{\LV(p')}\\
            {\LV(\ul\Gamma)}\arrow[swap]{r}{\psi}&{\LV(\ul{\Gamma'})}
    \end{tikzcd}
  \end{equation}
  If $\phi$ is an LV morphism, then $\psi$ is an LV morphism.
\end{lemma}
\begin{proof}
Suppose that $\phi$ is an LV morphism and consider the map $\varphi:=\LV(p')\circ\phi=\psi\circ\LV(p)$, which is an
LV morphism.  For any functions $F,G$ on $\LV(\ul{\Gamma'})$ (i.e., on $\bbF[\ul{S'}]$, with
$\ul{\Gamma'}=(\ul{S'},\ul{A'})$),
\begin{align*}
  \LV(p)^*\psi^*\pb{F,G}_{\ul{A'}}&=\varphi^*\pb{F,G}_{\ul{A'}}=\pb{\varphi^*F,\varphi^*G}_{A}\;;\\
  \LV(p)^*\pb{\psi^*F,\psi^*G}_{\ul{A}}&=\pb{\LV(p)^*\psi^*F,\LV(p)^*\psi^*G}_{A}=\pb{\varphi^*F,\varphi^*G}_{A}\;,
\end{align*}
where we have used that $\varphi$ and $\LV(p)$ are Poisson morphisms. Since $\LV(p)$ is surjective, $\LV(p)^*$ is
injective and we can conclude that $\psi^*\pb{F,G}_{\ul{A'}}=\pb{\psi^*F,\psi^*G}_{\ul{A}}$. This shows that $\psi$
is a Poisson map. Similarly, since $\varphi$ and $\LV(p)$ preserve the respective Hamiltonians,
\begin{equation*}
  \LV(p)^*\psi^* H_{\ul{S'}}=\varphi^*H_{\ul{S'}}=H_S=\LV(p)^*H_{\ul S}\;,
\end{equation*}%
we may conclude from the injectivity of $\LV(p)^*$ that $\psi^*H_{\ul{S'}}=H_{\ul S}$.
This shows that $\psi$ is an LV morphism.
\end{proof}
It is easy to construct a counterexample for the inverse implication of the previous lemma, though the analogous
property for graphs is an equivalence (see Lemma \ref{lma:graph_morphisms}).

We now show that surjective LV morphisms can be decloned, just like surjective graph morphisms (Proposition
\ref{prp:graph_decloning}). Notice that, again, weighted LV systems are not considered here, because we don't have a
notion of morphism between such systems.
\begin{prop}\label{prp:LV_decloning}
  Suppose that $\phi:\LV(\Gamma)\to\LV(\Gamma')$ is a surjective LV morphism. Denote by $\ul\Gamma$ and
  $\ul{\Gamma'}$ the decloned graphs of $\Gamma$ and~$\Gamma'$.
  \begin{enumerate}
    \item The LV morphism $\phi$ induces a unique LV morphism $\ul\phi:\LV(\ul{\Gamma})\to\LV(\ul{\Gamma'})$ such that
      the following diagram of (surjective) LV morphisms is commutative:
        \begin{equation}\label{dia:LV}
          \begin{tikzcd}[row sep=5ex, column sep=7ex]
            \LV(\Gamma)\arrow{r} {\phi}\arrow[swap]{d} {\LV(p)}&\LV(\Gamma')\arrow{d}{\LV(p')}\\
             \LV(\ul\Gamma)\arrow[swap]{r}{\ul\phi}& \LV(\ul{\Gamma'})
          \end{tikzcd}
        \end{equation}
  \item If $\phi':\LV(\Gamma')\to\LV(\Gamma'')$ is another surjective LV morphism, then
    $\ul{\phi'\circ\phi}=\ul{\phi'}\circ\ul\phi$; also,
    $\ul{\Id_{\LV(\Gamma)}}=\Id_{\LV(\ul\Gamma)}$.
  \end{enumerate}
\end{prop}
\begin{proof}
Consider again the map $\varphi:=\LV(p')\circ\phi$, which is a surjective LV morphism with irreducible target. By
item (1) in Proposition \ref{prp:normal_form}, $\varphi(\a_s)_{s\in S}=\(\sum_{s\in S_{\ul u}}\a_s\)_{\ul
  u\in\ul{S'}}$, for all $(\a_s)_{s\in S}\in\bbF[S]$ and for some partition $(S_{\ul u})_{\ul u\in{\ul{S'}}}$ of
$S$.  Suppose that $s\sim t$ with $s\in S_{\ul u}$ and $t\in S_{\ul v}$. We show by contradiction that $\ul u=\ul
v$.  Assume therefore that $\ul u\neq\ul v$, still assuming that $s\sim t$.  In view of item (2) of Proposition
\ref{prp:normal_form}, $a_{q,r}=\ul a'_{\ul u,\ul v}$, independently of $q\in S_{\ul u}$ and $r\in S_{\ul v}$;
since $a_{s,t}=0$ (because $s\sim t$), $\ul a'_{\ul u,\ul v}=0$. This shows that $\ul a'_{\ul u,\ul w}=\ul a'_{\ul
  v,\ul w}$ when $\ul w=\ul u$ or $\ul w=\ul v$. For $\ul w\in\ul{S'}$ different from $\ul u$ and $\ul v$ and $r\in
S_{\ul w}$, again by the cited item, and since $s\sim t$,
\begin{equation*}
  \ul a'_{\ul u,\ul w}=a_{s,r}=a_{t,r}=\ul a'_{\ul v,\ul w}\;,
\end{equation*}%
so that $\ul a'_{\ul u,\ul w}=\ul a'_{\ul v,\ul w}$ for all $\ul w\in\ul{S'}$, which shows that $\ul
u\sim\ul{v}$. Since $\ul{\Gamma'}$ is irreducible, this implies that $\ul u=\ul v$, which contradicts
the assumption that $\ul u\neq\ul v$. Therefore $\ul u=\ul v$. The partition of $S$, defined by $\sim$,
is therefore a refinement of the partition $(S_{\ul{u}})_{\ul u\in \ul{S'}}$.
This allows us to define a map
\begin{equation}\label{eq:def_varphi_cloned}
  \begin{array}{lcccl}
    \ul\phi&:&\bbF[\ul S]&\to&\bbF[\ul{S'}]\\
      & &(\gamma_{\ul s})_{\ul s\in \ul S} &\mapsto&
         \left(\sum_{\ul s\subset S_{\ul u}}\gamma_{\ul s}\right)_{\ul u\in \ul{S'}}\;.
  \end{array}
\end{equation}
%
We show that $\ul\phi$ makes (\ref{dia:LV}) into a commutative diagram. Let $(\a_s)_{s\in
  S}\in\bbF[S]$. Then, by (\ref{eq:LV(p)}), the definition (\ref{eq:def_varphi_cloned}) of $\ul\phi$ and
(\ref{eq:normal_form}), in that order, we get
\begin{align*}
  \(\ul\phi\circ\LV(p)\)(\a_s)_{s\in S}&=\ul\phi\(\sum_{s\in\ul s}\a_s\)_{\ul s\in\ul S}
  =\(\sum_{\ul s\subset S_{\ul u}}\sum_{s\in\ul s}\a_s\)_{\ul u\in\ul {S'}}\\
  &=\(\sum_{s\in S_{\ul u}}\a_s\)_{\ul u\in\ul {S'}}=\varphi(\a_s)_{s\in S}=(\LV(p')\circ\phi)(\a_s)_{s\in S}\;.
\end{align*}
It follows that $\ul\phi\circ\LV(p)=\LV(p')\circ\phi$, as was to be shown. In view of Lemma~\ref{lma:LV_morphisms},
$\ul\phi$ is an LV morphism. This proves item~(1). Item (2) follows at once from it by the uniqueness of $\ul\phi$.
\end{proof}
The linear map $\ul\phi$, induced by $\phi$, is called its \emph{decloned LV morphism}. It is clear that the above
proposition says that decloning of LV systems is a functor, just like decloning of graphs. We will come back to
this in Section~\ref{par:LV_cat}.

\subsection{Isomorphisms of Lotka-Volterra systems}\label{par:LV_isomorphisms}
We now consider LV isomorphisms. We have already seen in Figure \ref{fig:cex} that not all morphisms of LV systems
are induced by graph morphisms, even when the underlying graphs are irreducible. However, isomorphisms (and in
particular automorphisms) of irreducible LV systems are induced by graph morphisms, as we show in the following
proposition:
\begin{prop}\label{prp:iso_irr}
  Let $\Gamma$ and $\Gamma'$ be two skew-symmetric graphs, with $\Gamma'$ assumed irreducible.  If
  $\phi:\LV(\Gamma)\to\LV(\Gamma')$ is an LV isomorphism, then $\phi=\LV(\Phi)$ for a unique graph isomorphism
  $\Phi:\Gamma\to\Gamma'$.
%
%
\end{prop}
\begin{proof}
Suppose that $\phi:\LV(\Gamma)\to\LV(\Gamma')$ is an LV isomorphism. According to Proposition
\ref{prp:normal_form}, $\phi$ is of the form $\phi\(\a_s\)_{s\in{S}}=\(\sum_{s\in{S_u}}\a_s\)_{u\in{S'}}$, where
the subsets $S_u$ form a partition of $S$, indexed by $S'$. Since $\#S=\#S'$ every part $S_u$ of $S'$ must be a
singleton, and we can define a bijection $\Phi:\Gamma\to\Gamma'$ by letting $S_u=\set{\Phi(u)}$. Then $\phi$ takes
the simple form $\phi\(\a_s\)_{s\in{S}}=\(\a_{\Phi(u)}\)_{u\in{S'}}$, i.e., $\phi$ simply permutes the coordinates,
as dictated by $\Phi$ and $\phi^*y_u=x_{\Phi(u)}$ for all $u\in S'$. Let
$u,v\in S'$ and denote $s:=\Phi(u)$ and $t:=\Phi(v)$. Then
\begin{align*}
  a_{s,t}x_sx_t&=\pb{x_s,x_t}_A=\pb{\phi^*y_u,\phi^*y_v}_{A}=\phi^*(\pb{y_u,y_v}_{A'})=a'_{u,v}x_sx_t\;,
\end{align*}
which shows that $\Phi$ is a graph morphism; $\phi$ is the linear extension of $\Phi$, hence
$\phi=\LV(\Phi)$. Since the partition in subsets $S_u$ is uniquely determined by $\phi$, the map~$\Phi$ is unique.
\end{proof}
Thanks to Proposition \ref{prp:iso_irr}, we can define the notion of a weighted LV isomorphism between
irreducible LV systems, and show that the graph underlying an LV system is unique, up to isomorphism.
\begin{defn}\label{def:iso_weighted}
  Let $(\Gamma,\wght)$ and $(\Gamma',\wght')$ be weighted irreducible graphs. Let $\psi:\LV(\Gamma)\to\LV(\Gamma')$
  be an LV isomorphism and let $\Psi:\Gamma\to\Gamma'$ denote the graph isomorphism for which
  $\psi=\LV(\Psi)$. Then $\psi:(\LV(\Gamma),\wght)\to(\LV(\Gamma'),\wght')$ is said to be an \emph{isomorphism of
    weighted Lotka-Volterra systems}, or simply a \emph{weighted LV isomorphism} if $\Psi^*\wght'=\wght$, i.e.,
  when $\Psi:(\Gamma,\wght)\to(\Gamma',\wght')$ is a weighted graph isomorphism.
\end{defn}
\begin{prop}\label{prop:iso}
  Let $\Gamma$ and $\Gamma'$ be skew-symmetric graphs. Then $\Gamma\simeq\Gamma'$ if and only if
  $\LV(\Gamma)\simeq\LV(\Gamma')$.
\end{prop}
\begin{proof}
According to Proposition \ref{prp:functor}, we only need to show the inverse implication. As before, we write
$\Gamma=(S,A)$ and $\Gamma'=(S',A')$, and we denote their weighted decloned graphs by $(\ul\Gamma,\wght_\Gamma)$
and $(\ul{\Gamma'},\wght_{\Gamma'})$. The coordinates on $\bbF[S]$ and on $\bbF[S']$ are respectively denoted by
$x_s$, with $s\in S$ and $y_u$, with $u\in S'$. Suppose that $\phi:\LV(\Gamma)\to\LV(\Gamma')$ is an LV
isomorphism.  We know from Lemma \ref{lma:pppambos} that if $u, v\in S'$, with $u\not\sim v$, then the functions
$x_s$ of which $\phi^*y_u$ and $\phi^*y_v$ are a linear combination are all different. Let $\ul u\in \ul{S'}$ and
denote its clones by $u_1,u_2,\dots,u_{\wght_{\Gamma'}({\ul u})}$. Then
$\phi^*y_{u_1},\phi^*y_{u_2},\dots,\phi^*y_{u_{\wght_{\Gamma'}({\ul u})}}$ depend only on a certain set of coordinate
functions $x_s$; since $\phi$ is an isomorphism, their number is also $\wght_{\Gamma'}({\ul u})$, so say they are
$x_{s_1},x_{s_2},\dots,x_{s_{\wght_{\Gamma'}({\ul u})}}$.  The notation suggests that the vertices $s_i$ are all
equivalent, so that $\wght_{\Gamma'}(\ul u)\leqs\wght_\Gamma({\ul s})$ and this is the case. Indeed, since the
above functions $y_{u_i}$ are in involution, the same is true for the functions $x_{s_i}$, and the claim follows
from item (2) in Proposition \ref{prp:normal_form}. We therefore get a bijection $\Psi:\ul S\to\ul{S'}$ such that
for any $\ul s\in\ul S$, only the functions $y_u$, with $u$ being a clone of $\Psi(\ul s)$, depend on the functions
$x_s$, with~$s$ a clone of $\ul s$, and $\Psi$ satisfies
$\wght_{\Gamma'}({\Psi(\ul{s})})\leqs\wght_{\Gamma}({\ul{s}})$ for all $\ul s\in\ul S$; since
$\vert\wght_{\Gamma'}\vert=\vert\wght_{\Gamma}\vert$ (because $\phi$ is an isomorphism), we must have equality,
$\wght_{\Gamma'}({\Psi(\ul s)})=\wght_{\Gamma}({\ul s})$ for all $\ul s\in\ul S$. The isomorphism~$\ul\phi$, induced
by~$\phi$, as given by Proposition \ref{prp:LV_decloning}, is just the linear extension of $\Psi$. Therefore,
$\ul\phi:(\LV(\ul\Gamma),\wght_\Gamma)\to(\LV(\ul{\Gamma'}),\wght_{\Gamma'})$ is a weighted LV isomorphism, induced
by a weighted graph isomorphism $\Psi:(\ul\Gamma,\wght_\Gamma)\to(\ul{\Gamma'},\wght_{\Gamma'})$. By item (2) in
Proposition \ref{prp:graph_autom_decloning}, $\Gamma$ and $\Gamma'$ are isomorphic, as was to be shown.
\end{proof}
Combined with Proposition \ref{prp:diffeo_to_linear} and item (2) of Proposition \ref{prp:graph_autom_decloning},
we get the following result on the classification of LV systems:
\begin{thm}
  Let $\Gamma$ and $\Gamma'$ be two skew-symmetric graphs. The following are equivalent:
  \begin{enumerate}
    \item The LV systems $\LV(\Gamma)$ and $\LV(\Gamma')$ are smoothly isomorphic;
    \item The LV systems $\LV(\Gamma)$ and $\LV(\Gamma')$ are linearly isomorphic;
    \item The graphs $\Gamma$ and $\Gamma'$ are isomorphic;
    \item The weighted (irreducible) graphs $(\ul\Gamma,\wght_\Gamma)$ and $(\ul{\Gamma'},\wght_{\Gamma'})$ are isomorphic.
  \end{enumerate}
  \qed
\end{thm}
The classification of LV systems, modulo smooth isomorphisms, is therefore the same as the
classification of (weighted irreducible) graphs, modulo (weighted) graph isomorphisms.

\subsection{Automorphisms of Lotka-Volterra systems}\label{par:LV_automorphisms}

We now turn to LV automorphisms. We denote, for an irreducible weighted graph $(\Gamma,\wght)$, by
$\Aut(\LV(\Gamma),\wght)$ the group of automorphisms of the weighted LV system $(\LV(\Gamma),\wght)$. It is a
subgroup of $\Aut(\LV(\Gamma))$. Since $\Phi^*\wght_0=\wght_0$ for any $\Phi\in\Aut(\Gamma)$, the groups
$\Aut(\LV(\Gamma),\wght_0)$ and $\Aut(\LV(\Gamma))$ are isomorphic.  Moreover, Proposition \ref{prp:iso_irr} also
leads to the following description of the automorphism group of an irreducible LV system.

\begin{prop}\label{prp:iso_irred}
  Let $(\Gamma,\wght)=(S,A)$ be a weighted graph, which is assumed irreducible.
  \begin{enumerate}
    \item If $\phi\in\Aut(\LV(\Gamma))$, then $\phi=\LV(\Phi)$ for a unique $\Phi\in\Aut(\Gamma)$.
    \item The map $\LV:\Aut(\Gamma)\to\Aut(\LV(\Gamma))$ is a group isomorphism,
      as well as its restriction $\LV_0:\Aut(\Gamma,\wght)\to\Aut(\LV(\Gamma),\wght)$.
  \end{enumerate}
\end{prop}
\begin{proof}
Item (1) is a particular case of Proposition \ref{prp:iso_irr}. According to item~(3) of Proposition
\ref{prp:functor}, the restriction of the functor $\LV$ to $\Aut(\Gamma)$ (which we still denote by $\LV$) takes
values in $\Aut(\LV(\Gamma))$. Item (2) of the same proposition says that this restriction is a group
homomorphism. According to item (1) it is bijective, hence $\LV:\Aut(\Gamma)\to\Aut(\LV(\Gamma))$ is a group
isomorphism. Since the condition of preserving the weight $\wght$ is by definition the same for a weighted LV
isomorphism as for a weighted graph morphism, the isomorphism $\LV$ further restricts to a group isomorphism
$\LV_0:\Aut(\Gamma,\wght)\to\Aut(\LV(\Gamma),\wght)$.
\end{proof}
When $\Gamma$ is reducible, the group of automorphisms of $\LV(\Gamma)$ is much larger than $\Aut(\Gamma)$. It is
in fact infinite, as we show in the following proposition:

\begin{prop}\label{prp:many_auto}
  Let $(\Gamma,\wght)$ be a weighted graph, with $\Gamma=(S,A)$, and with cloned graph
  $\Gamma^\wght=(S^\wght,A^\wght)$. Suppose that $\phi$ is a vector space automorphism of~$\bbF[S^\wght]$, leaving
  for every $s\in S$ the subspace $\Span\set{s_i\mid 1\leqslant i\leqslant \wght(s)}$ invariant; suppose also that
  $\phi$ leaves the sum of the associated coordinates~$x_{s_i}$ invariant, i.e.,
  $\phi^*\sum_{i=1}^{\wght(s)}x_{s_i}=\sum_{i=1}^{\wght(s)}x_{s_i}.$ Then $\phi$ is an LV automorphism
  of~$\LV(\Gamma^\wght)$.
\end{prop}
\begin{proof}
Let $\phi$ be a vector space automorphism of~$\bbF[S^\wght]$ and suppose that $\phi$ leaves for every $s\in S$ the
subspace $\Span\set{s_i\mid 1\leqslant i\leqslant \wght(s)}$ invariant. Then, for any $s\in S$ and
$i=1,2,\dots,\wght(s)$,
\begin{equation*}
  \phi^*x_{s_i}=\sum_{j=1}^{\wght(s)}\b_{i,j}^{(s)}x_{s_j}\;, \qquad \text{where}\;
  \(\b_{i,j}^{(s)}\)\in\GL(\wght(s),\bbF)\;.
\end{equation*}
The condition that $\phi^*\sum_{i=1}^{\wght(s)}x_{s_i}=\sum_{i=1}^{\wght(s)}x_{s_i}$ amounts to
$\sum_{i=1}^{\wght(s)}\b_{i,j}^{(s)}=1$ for $j=1,\dots,\wght(s)$ and therefore $\phi^*H_S=H_S$.
Let us denote the linear coordinates on $\bbF[S]$ by $x_s$, where $s\in S$.
For $s,t\in S$ and $i\in\set{1,2,\dots,\wght(s)},\; j\in\set{1,2,\dots,\wght(t)}$, because
$\pb{x_{s_i},x_{t_j}}_{A^\wght}=a_{s,t}x_{s_i}x_{t_j}$, we have
\begin{align*}
\pb{\phi^*x_{s_i},\phi^*x_{t_j}}_{A^\wght}
&=\sum_{k=1}^{\wght(s)}\sum_{\ell=1}^{\wght(t)}\b_{i,k}^{(s)}\b_{j,\ell}^{(t)}a_{s,t}x_{s_k}x_{t_\ell}\\
&=a_{s,t}\phi^*x_{s_i}\phi^*x_{t_j}=\phi^*\pb{x_{s_i},x_{t_j}}_{A^\wght}\,,
\end{align*}
and therefore $\phi$ is indeed an LV automorphism of $\LV(\Gamma^\wght)$.
\end{proof}
Propositions \ref{prp:iso_irred} and \ref{prp:many_auto} lead at once to the following corollary:
\begin{cor}
  Let $\Gamma$ be a skew-symmetric graph. The LV system $\LV(\Gamma)$ is irreducible if and only if its
  automorphism group $\Aut(\LV(\Gamma))$ is finite.\qed
\end{cor}

We denote the group of matrices $\(\a_{i,j}\)\in\GL(n,\bbF)$ for which $\sum_{i=1}^{n}\a_{i,j}=1$ for $j=1,\dots,n$
by $\GLp(n,\bbF)$. Since it is the group of invertible matrices fixing a non-zero vector, it is isomorphic to an
affine group,
\begin{equation*}
  \GLp(n,\bbF)\simeq\Aff(n-1,\bbF)\simeq\bbF^{n-1}\rtimes\GL(n-1,\bbF)\;.
\end{equation*}%
In the notation of Proposition \ref{prp:many_auto}, it follows from the proposition that the LV morphism associated
to a permutation of the clones of a vertex $s\in S$ is an element of $\GLp(\wght(s),\bbF)$, viewed as an
automorphism of $\LV(\Gamma)$.
\begin{prop}
  Suppose that $\Gamma=(S,A)$ is a skew-symmetric graph, and denote its weight\-ed decloned graph by
  $(\ul\Gamma,\wght_\Gamma)$, with $\ul\Gamma=(\ul S,\ul A)$. The following diagram is commutative and its lines
  are split short exact sequences.
  \begin{equation*}
    \begin{tikzcd}[row sep=5ex, column sep=4ex]
      0\arrow{r} {}&\displaystyle\prod_{\ul s\in\ul S}\mathcal{S}_{\wght_\Gamma({\ul s})}\arrow{r}{}\arrow{d} {}&
     \Aut(\Gamma)\arrow{r}{}\arrow{d}{\LV}&\Aut(\ul\Gamma,\wght_\Gamma)\arrow{r}{}\arrow{d}{\LV_0}\arrow[swap]{d}{\simeq}&0\\
      0\arrow{r} {}&\displaystyle\prod_{\ul s\in\ul S}\GLp(\wght_\Gamma({\ul s}),\bbF)\arrow{r}{}&
        \Aut(\LV(\Gamma))\arrow{r}{}&\Aut(\LV(\ul\Gamma),\wght_\Gamma)\arrow{r}{}&0
    \end{tikzcd}
  \end{equation*}
  As a consequence, $\Aut(\LV(\Gamma))$ is a semi-direct product,
  \begin{equation*}
    \Aut(\LV(\Gamma))\simeq\prod_{\ul s\in\ul S}\GLp(\wght_\Gamma({\ul s}),\bbF)\rtimes\Aut(\ul\Gamma,\wght_\Gamma)\;.
  \end{equation*}%
\end{prop}
\begin{proof}
  Since $\ul\Gamma^{\wght_\Gamma}\simeq\Gamma$ in a canonical way, we can identify the vertex set $S$ of $\Gamma$
  with the vertex set $\ul S^{\wght_\Gamma}$ of $\ul\Gamma^{\wght_\Gamma}$ and so $S=\set{\ul s_i\mid\ul s\in\ul S,\
    i=1,2,\dots,\wght_\Gamma(\ul s)}$.  The left square of the diagram is commutative because the leftmost vertical
  arrow in the diagram is the restriction of $\LV$ to a subgroup. The right square is commutative because of the
  definition of weighted decloning of LV automorphisms, in case the automorphism is induced by a graph
  automorphism. We have already shown in Proposition \ref{prp:graph_autom} that the top line is a split short exact
  sequence. It implies, by commutativity of the diagram, that
  $\Aut(\LV(\Gamma))\to\Aut(\LV(\ul\Gamma),\wght_\Gamma)$ is surjective. The kernel of this map consists of the
  automorphisms of~$\bbF[\ul S^{\wght_\Gamma}]$, leaving for every $\ul s\in \ul S$ the subspace $\Span\set{\ul
    s_i\mid 1\leqslant i\leqslant \wght_\Gamma(\ul s)}$ invariant, as well as the sum of the associated coordinates
  $x_{\ul s_i}$, i.e., $\phi^*\sum_{i=1}^{\wght_\Gamma(\ul s)}x_{\ul s_i}= \sum_{i=1}^{\wght_\Gamma(\ul s)}x_{\ul
    s_i}$, so it is precisely $\prod_{\ul s\in\ul S}\GLp(\wght_\Gamma({\ul s}),\bbF)$. The commutativity of the
  diagram and the fact that $\LV_0$ is an isomorphism imply that, since the upper line is split, the lower line is
  also split.
\end{proof}

\subsection{Functorial interpretation}\label{par:LV_cat}
We have already given a functorial interpretation of the results obtained in Section \ref{sec:graphs} (see Section
\ref{par:cat_graphs}). We now complete this with the results of Section \ref{sec:LV}. Consider the following
diagram:
\begin{equation*}
  \begin{tikzcd}[row sep=6ex, column sep=6ex]
    \Gr\arrow{r}{\rho}\arrow[swap]{d} {\LV}&\Grz\arrow{d}{\LV_0}\arrow[shift left=1ex] {l}{\imath_0}\\
    \bLV\arrow[swap]{r}{\sigma}&\bLVz\arrow[swap,shift right=1ex] {l}{\jmath_0}
  \end{tikzcd}
\end{equation*}
Its top line has already been discussed in Section \ref{par:cat_graphs}: $\Gr$ is the category of skew-symmetric
graphs with surjective morphisms and $\Grz$ is the full subcategory of irreducible graphs, with inclusion functor
$\imath_0$. Also, $\rho$ is the decloning functor. We define similarly $\bLV$ to be the category whose objects are
LV systems $(\bbF[S],A,H_S)$ and whose morphisms are surjective LV morphisms, i.e., surjective linear Poisson maps,
which preserve the Hamiltonian. The full subcategory of irreducible LV systems is denoted $\bLVz$. According to
Propositions~\ref{prp:decloning} and~\ref{prp:LV_decloning}, decloning of LV systems defines a functor $\sigma$,
which is according to item (1) in Proposition \ref{prp:LV_decloning} a reflection functor for the inclusion functor
$\jmath_0:\bLVz\to\bLV$. In particular, $\sigma$ is an adjoint functor for
$\jmath_0$. Proposition~\ref{prp:functor} says that we have a functor $\LV:\Gr\to\bLV$ which associates to a
skew-symmetric graph $\Gamma$ the corresponding LV system $\LV(\Gamma)$, and to a surjective graph morphism
$\Phi:\Gamma\to\Gamma'$ the LV morphism $\LV(\Phi):\LV(\Gamma)\to\LV(\Gamma')$. It is an embedding functor and its
restriction to $\Gr_0$ is the functor $\LV_0$, which takes by definition values in~$\bLVz$. According to
Proposition \ref{prop:iso} the functor $\LV$, and therefore also its restriction $\LV_0$, is a conservative functor
(it reflects isomorphisms).

The commutativity of the diagram, to wit the fact that $\LV_0\circ\rho=\s\circ\LV$ and
$\LV\circ\,\imath_0=\jmath_0\circ\LV_0$, is clear on objects; on morphisms, it follows for the first equality from
the fact that the decloning of a surjective LV morphism, induced by a graph morphism, is just the decloned graph
morphism, viewed as an LV morphism; on morphisms, the second equality follows at once from the fact that $\LV_0$ is
a restriction of the functor $\LV$ to the subcategories $\Grz$ and $\bLVz$ of $\Gr$ and $\bLV$.

\subsection{Population dynamics interpretation of (de-)cloning}
Lotka-Vol\-terra systems first appeared in the context of population dynamics \cite{Lotka,Volterra}. In this
context, $\bbF=\bbR$, the set $S$ is the set of species and for each $s\in S$ the function $x_s$ is the number of
individuals belonging to the species $s$. The differential equations
\begin{equation*}
  \dot{x}_{s} = \sum_{t\in S} a_{s,t}\,x_{s}x_{t}\;, \qquad \text{for all}\; s\in S\;,
\end{equation*}
describe the evolution of the number of individuals of each species; the real parameter $a_{s,t}$ governs the
interaction between the species $s$ and $t$, with $a_{s,t}$ being positive (respectively negative) meaning that the
number of individuals of species $s$ will grow (respectively shrink) proportionally to the number of individuals of
species $s$, to the number of individuals of species $t$, and to $a_{s,t}$. In this model, cloning is the natural
procedure of subdividing each species in subspecies, the interaction between the subspecies of two different
species being the same as the interaction between the original species, with no interaction between two subspecies
of the same species. Also here, decloning is more important than cloning, since it amounts to simplifying the model
by assembling similar species in a single one. From that point of view, the decloning map $\LV(p)$ of Proposition
\ref{prp:decloning} is crucial.

We show in the next proposition that under the reasonable assumption that the interaction between two species does
not depend on the other species, Proposition \ref{prp:decloning} only holds when these interactions are quadratic,
with linear contribution by each of the two species, i.e., are Lotka-Volterra models.

\begin{prop}\label{prp:LV_char}
  Let $S$ be a finite set and suppose that $\pi$ is a Poisson structure on $\bbF[S]$ for which
  $\pb{x_s,x_t}=\pi_{s,t}=\pi_{s,t}(x_s,x_t)$, i.e., the function $\pi_{s,t}$ depends on $x_s$ and $x_t$ only. Let
  $\wght:S\to\bbN^*$ be a weight vector and consider the bivector field $\ol\pi$ on $\bbF[S]$ defined by setting
  $\pb{x_{s_i},x_{t_j}}_{\ol\pi}:= \pi_{s,t}(x_{s_i},x_{t_j})$. Suppose that $\ol\pi$ is a Poisson structure and
  that the decloning map $\LV(p)$ is a Poisson map. Then, for every $s\in S$ for which $\wght(s)>1$, $\pi_{s,t}$ is a
  linear function of $x_s$. In particular, if $\wght(s)>1$ and $\wght(t)>1$ then $\pi_{s,t}=a_{s,t}x_sx_t$ for some
  $a_{s,t}\in\bbF$, and if all $\wght(s)>1$ then $\pi$ is a diagonal Poisson structure and $(\bbF[S],\pi,H_S)$ is an
  LV system, with $\pi=\pi_A$, where $A=(a_{s,t})$.
\end{prop}
The proof which we will give is based on the following lemma:
\begin{lemma}\label{lma:add_to_linear}
  Let $F(\a,\b)$ be a non-zero function in two variables. In case $\bbF=\bbR$, we assume $F$ to be smooth; if
  $\bbF=\bbC$, then $F$ is assumed to be holomorphic. Let $m,n\in\bbN$ with $m>1$ and $n>0$. If
  \begin{equation}\label{eq:F_cond}%
    \sum_{s=1}^m\sum_{t=1}^n F(\a_s,\b_t)=F\(\sum_{s=1}^m \a_s,\sum_{t=1}^n \b_t\)\;
        \quad \text{for all}\; \a_1,\dots,\a_m,\b_1,\dots,\b_n\in\bbF\;,
  \end{equation}%
  then
  \begin{equation}\label{eq:F_prop}%
    F(\a+\a',\b)=F(\a,\b)+F(\a',\b) \quad \text{for all}\; \a,\a',\b\in\bbF\;,
  \end{equation}%
  so that $F$ is linear in its first argument, $F(\a,\b)=\a G(\b)$ for some function~$G$ (in one variable).
\end{lemma}
\begin{proof}
Let us first recall the standard argument that (\ref{eq:F_prop}) implies that $F$ is linear in its first
argument. The condition says that $F$ is additive in its first argument, so that $F(\l \a,\b)=\l F(\a,\b)$ for all
$\l\in\bbQ$ and $\a,\b\in\bbF$. By continuity of $F$, this holds for all $\l\in\bbR$. In particular, $F(0,\b)=0$ for
all $\b\in\bbF$. It remains to be shown that $F$ is complex linear when $\bbF=\bbC$. Fix $\b$ and write
$K(\a):=F(\a,\b)$. Write the power series expansion of $K$ around~$0$, $K(\a)=\sum_{i>0}c_i\a^i$. Then, for any
$\l\in\bbR$,
\begin{equation*}%
  \sum_{i>0}c_i\l^i\a^i=K(\l \a)=\l K(\a)=\l\sum_{i>0}c_i\a^i\;,
\end{equation*}%
so that $\sum_{i>1}c_i(\l^i-\l)\a^i=0$ for all $\a$ around $0$. Taking $\l$ real, but different from $0,\pm1$, it
follows that $c_i=0$ for $i>1$, so that $K(\a)=c_1\a$ is linear in~$\a$.

We now prove (\ref{eq:F_prop}) when $m,n>1$. First notice that, taking all $\a_s$ and~$\b_t$ equal to zero in
(\ref{eq:F_cond}), we find $mn F(0,0)=F(0,0)$, so that $F(0,0)=0$. Next, take $\a_2,\dots,\a_m,\b_1,\dots,\b_n$
equal to zero in (\ref{eq:F_cond}) to find $n F(\a_1,0)=F(\a_1,0)$, so that $F(\a,0)=0$ for all $\a$; similarly,
$F(0,\b)=0$ for all $\b$. Finally take $\a_3,\dots,\a_m,\b_2,\dots,\b_n$ equal to zero in (\ref{eq:F_cond}) to find
$F(\a_1,\b_1)+F(\a_2,\b_1)=F(\a_1+\a_2,\b_1)$ for all $\a_1,\a_2,\b_1\in\bbF$, which is (\ref{eq:F_prop}).

To finish, we prove (\ref{eq:F_prop}) when $m>1$ and $n=1$. The condition (\ref{eq:F_cond}) simplifies to
\begin{equation}\label{eq:F_cond_simp}%
  \sum_{s=1}^m F(\a_s,\b)=F\(\sum_{s=1}^m \a_s,\b\)\;  \qquad \text{for all}\; \a_1,\dots,\a_m,\b\in\bbF\;.
\end{equation}
When $m=2$ the conditions (\ref{eq:F_prop}) and (\ref{eq:F_cond_simp}) are identical, so let us suppose that
$m>2$. Taking all $\a_s$ equal to zero in (\ref{eq:F_cond_simp}) we find $m F(0,\b)=F(0,\b)$, so $F(0,\b)=0$ for
all $\b\in\bbF$ (since $m\neq1$). Take now only $\a_s=0$ for $s>2$ in (\ref{eq:F_cond_simp}) to find
$F(\a_1,\b)+F(\a_2,\b)=F(\a_1+\a_2,\b)$ for all $\a_1,\a_2,\b\in\bbF$, which is again (\ref{eq:F_prop}).
\end{proof}

\begin{proof} [Proof (of Proposition \ref{prp:LV_char})]\ \
Suppose that $\wght(s)>1$ for some $s\in S$, and let $t\in S$ with $t\neq s$. Since $\LV(p)$ is assumed to be a
Poisson map,
\begin{align*}
  \sum_{i=1}^{\wght(s)}\sum_{j=1}^{\wght(t)}\pi_{s,t}(x_{s_i},x_{t_j})
    &=\pb{x_{s_1}+x_{s_2}+\cdots+x_{s_{\wght(s)}},x_{t_1}+x_{t_2}+\cdots+x_{t_{\wght(t)}}}_{\ol\pi}\\
    &=\pb{\LV(p)^*x_s,\LV(p)^*x_t}_{\ol\pi}=\LV(p)^*\(\pb{x_s,x_t}\)=\\
    &=\LV(p)^*\pi_{s,t}=\pi_{s,t}\left(\sum_{i=1}^{\wght(s)}x_{s_i},\sum_{j=1}^{\wght(t)}x_{t_j}\right)\;.
\end{align*}%
It now suffices to apply Lemma \ref{lma:add_to_linear} with $F=\pi_{s,t}$ and $m=\wght(s)$ and $n=\wght(t)$ to
conclude that $\pi_{s,t}$ is a linear function of $x_s$. When also $\wght(t)>1$, then $\pi_{s,t}$ is a linear
function of its two arguments $x_s$ and $x_t$, so $\pi_{s,t}=a_{s,t}x_sx_t$ for some $a_{s,t}\in\bbF$. This implies
that $\pi$ is a diagonal Poisson structure if all weights $\wght(s)$ are at least $2$.
\end{proof}
\section{Integrability and Lax equations}\label{sec:int_lax}

We have seen in Section \ref{sec:LV} that cloning and decloning play a special role in the classification of
Lotka-Volterra systems and in the description of their automorphism groups. In this section, we study cloning and
decloning of LV systems from the point of view of integrability and of Lax equations. In Section
\ref{par:examples} we recall the main examples of integrable LV systems. We show in Section
\ref{par:integrability} that cloning and decloning preserve integrability, and we construct in Section
\ref{subsec:lax} a Lax equation for the cloning of a large class of important examples, the Bogoyavlenskij systems
(recalled below).

\subsection{Integrable examples}\label{par:examples}
For a generic skew-symmetric graph $\Gamma$, the Hamiltonian system $\LV(\Gamma)$ is not integrable, though several
infinite families of LV systems are known to be integrable. Before giving a few examples of such families, we
recall two basic notions of integrability, adapted to the case of LV systems. Recall that for an $n$-dimensional LV
system $\LV(\Gamma)$, with $\Gamma=(S,A)$, the rank of the Poisson structure $\pi_A$ is the rank of the
skew-symmetric matrix $A$. Therefore, one needs for \emph{Liouville integrability} $n-\frac12\Rk A$ independent
first integrals, including the Hamiltonian, which are pairwise in involution (meaning that their Poisson bracket is
zero). We also recall that for \emph{superintegrability} one needs $n-1$ independent first integrals (the Poisson
structure does not intervene in the definition of superintegrability).

In the examples which follow, we give some families of LV systems which are Liouville or superintegrable
(conjecturally, for one of them). In each of the examples, we will use $S_n:=\set{1,2,\dots,n}$ as the vertex set
of the underlying graph.

\begin{example}\label{exa:KM}
The best known example of an integrable LV system is the $n$-particle periodic Kac-van Moerbeke system $\KM(n)$,
where $n\geqslant3$; it is the~LV system $\LV(\Gamma)$ where $\Gamma$ is a circuit with $n$ vertices. See
Figure~\ref{fig:1st}, where the graph underlying $\KM(6)$ is given. If we label the consecutive vertices in the
circuit as $1,2,\dots,n$, then the adjacency matrix $A$ of $\Gamma$ is given by
\begin{equation}\label{eq:A-matrix_KM}
  A=\left(
  \begin{array}{ccccrc}
    0&1& 0 &\cdots& 0 &-1\\
     -1 &0&1&0      &\dots   & 0\\
     0 & -1 &0   &      &   &\vdots\\
    \vdots&&\ddots&\ddots& &0\\
       0 &  & & & 0&1\\
    1&0&\cdots &\cdots&-1&0
  \end{array}
  \right)\,,
\end{equation}
and so the $\KM(n)$ vector field is given by
\begin{equation}\label{eq:KM-system}
  \dot{x}_{i}=x_{i}(x_{i+1}-x_{i-1})\;, \quad \text{for}\,\; i=1,\ldots,n\;,
\end{equation}
with the understanding that the indices are periodic modulo $n$, so that $x_{n+1}=x_1$ and $x_0=x_n$.  The rank of
$A$ is $n-1$ when $n$ is odd and $n-2$ otherwise. In the first case the product $x_1x_2\dots x_n$ is a Casimir
function; in the second case, both products $x_1x_3\dots x_{n-1}$ and $x_2x_4\dots x_n$ are Casimir functions. An
additional $\floor*{\frac{n-1}2}$ independent polynomial first integrals (including the Hamiltonian), in
involution, are constructed from a Lax equation (with spectral parameter) for $\KM(n)$, see \cite{KM_Prym} and
Section \ref{subsec:lax} below.  This accounts for the Liouville integrability of $\KM(n)$.

It is clear that a circuit with $n$ vertices is irreducible with automorphism group the cyclic group $C_n$. In view
of Proposition \ref{prp:iso_irred}, it follows that $\Aut(\KM(n))=C_n$; the only automorphisms of $\KM(n)$ are
those permutations of the variables $x_i$ which respect the cyclic order $1,2,\dots,n$.
\end{example}

\begin{example}\label{exa:aus}
An LV system $\LV(\Gamma)$ is said to be of \emph{maximal interaction} if its defining graph $\Gamma$ is a
tournament graph, i.e. every vertex is adjacent to all other vertices; said differently, the entries of the
adjacency matrix $A$ of~$\Gamma$ satisfy $a_{i,j}\neq0$ for $i\neq j$. A prime example is when $a_{i,j}=1$ for all
$i<j$,
\begin{equation}\label{eq:A-matrix_aus}
  A=\left(
  \begin{array}{ccccrc}
    0&1& 1 &\cdots& 1 &1\\
     -1 &0&1&1      &\dots   & 1\\
     -1 & -1 &0   &      &   &\vdots\\
    \vdots&&\ddots&\ddots& &1\\
       -1 &  & & & 0&1\\
    -1&-1&\cdots &\cdots&-1&0
  \end{array}
  \right)\;.
\end{equation}
We call the corresponding LV system $\LV(n,0)$, as we did in \cite{PPPP}. See Figures~\ref{fig:2nd} and
\ref{fig:3rd} which correspond to $n=6$ and $n=5$ respectively. For general $n$, the Hamiltonian vector field of
$\LV(n,0)$ is given by
\begin{equation}\label{eq:aus}
  \dot{x}_{i}=x_i\left(\sum_{j>i}x_{j}-\sum_{j<i}x_{j}\right)\;, \quad \text{for}\,\; i=1,\ldots,n\;.
\end{equation}
The rank of $A$ is $n-1$ when $n$ is odd; a rational Casimir function is then given by
\begin{equation*}
  C:=\frac{x_1x_3x_5\dots x_n}{x_2x_4\dots x_{n-1}}\;.
\end{equation*}%
When $n$ is even, the rank of $A$ is $n$ and so the constant functions are the only Casimir functions.  For any
$n$, the system possesses $n-1$ independent rational first integrals (including the Hamiltonian, which is
polynomial), making it superintegrable (see \cite{KKQTV}). Among these rational first integrals, which can be
constructed using Darboux polynomials, $\floor*{\frac{n+1}2}$ integrals in involution can be chosen, making the
system also Liouville integrable. See also \cite{KQV} for an integrable generalization.

The vertices of the graph $\Gamma$, underlying $\LV(n,0)$ all have a different outdegree, so $\Gamma$ is
irreducible and has trivial automorphism group. According to Proposition \ref{prp:iso_irred}, $\Aut(\LV(n,0))$ is
the trivial group.
\end{example}

\begin{example}\label{exa:bogo}
The $\KM(n)$ system has been generalized by Bogoyavlenskij \cite{Bog2} to the case of interaction between neighbors
at distance at most $k$, with $k<n/2$ (the $\KM(n)$ system corresponds to $k=1$). It is called the
\emph{Bogoyavlenskij system}, denoted $\B nk$. The defining graph $\Gamma_{n,k}=(S_n,A_{n,k})$ of $\B nk$ is the
graph with vertex set $S_n$ and with an arrow from $i$ to $j$ whenever $0< j-i\leqslant k$ or $0<i-j \leqslant
n-k$.  Said differently, there is an arrow from every vertex to the next $k$ vertices with respect to the cyclic
order on $S_n$. See for example Figure \ref{fig:1st} for the graph $\Gamma_{6,2}$. The adjacency matrix $A_{n,k}$
is the skew-symmetric Toeplitz matrix of order $n$, whose first line is given by
$$
  (0,\underbrace{1,1,\dots,1}_k,0,0,\dots,0,\underbrace{-1,-1,\dots,-1}_k)\;.
$$
The Hamiltonian vector field $\X_{H_{S_n}}$ takes the symmetric form
\begin{equation}\label{eq:bnk}
  \dot{x}_{i}=x_i\sum_{j=1}^k(x_{i+j}-x_{i-j})\;, \quad \text{for}\,\; i=1,\ldots,n\;,
\end{equation}
where the indices are again taken modulo $n$ and in $S_n$. Bogoyavlenskij gives in \cite{Bog2} a Lax pair $(L,M)$
with spectral parameter for $\B nk$, which we will recall and use in Section \ref{subsec:lax}. The coefficients of
the characteristic polynomial of $L$ are first integrals of $\B nk$, which have been shown to be in involution by
using the theory of $r$-matrices (see \cite{suris_book}). For small values of $n$ (and all $k<n/2$), one shows by
direct computation that this yields enough independent first integrals for Liouville integrability; conjecturally,
this holds for all $n$ (and all $k<n/2$).  For $k=1$ one gets $\B n1=\KM(n)$, which is Liouville integrable (see
Example \ref{exa:KM} above) and for $n=2k+1$, with $k$ arbitrary, one gets the so-called \emph{Bogoyavlenskij-Itoh
  system}, which is also known to be Liouville integrable (see \cite{Bog1,itoh1,itoh2}). In both cases, the first
integrals are obtained from the Lax operator $L$ in (\ref{eq:bogo_lax}). The graph $\Gamma_{n,k}$ is irreducible
and its automorphism group is the cyclic group $C_n$. It follows that the automorphism group of $\B nk$ is the
cyclic group $C_n$.
\end{example}

\begin{example}
As we already pointed out in Section \ref{par:basic}, LV systems can be reduced by setting one or several of the
variables $x_i$ equal to zero. In general, the reduced systems, which are still LV systems, may be integrable or
not. The LV system $\LV(n,0)$ (see Example~\ref{exa:aus}) can be seen as such a reduction of the
Bogoyavlenskij-Itoh system $\B{2n-1}{n-1}$, and it is Liouville integrable with rational first integrals; moreover
it is superintegrable. Reductions of $\LV(n,0)$ are of the form $\LV(m,0)$ with $m<n$ (see Figure \ref{fig:3rd}),
hence are also Liouville and superintegrable. For the more general systems $\LV(n,k)$, which are also integrable
reductions of the Bogoyavlenskij-Itoh system $\B{2n-2k-1}{n-k-1}$, see \cite{PPPP}. The reductions of $\KM(n)$ are
also Liouville integrable with polynomial first integrals since they are products of open (non-periodic) KM systems
as can again be easily seen from their underlying graphs: the underlying graph of the $n$-particle \emph{open} KM
system is the chain with vertices $\set{1,2,\dots,n}$ and an arrow from $i$ to $i+1$ for $i=1,2,\dots,n-1$. The
first integrals of the latter system are obtained by reduction from the first integrals of $\KM(n+1)$, in
particular they are polynomial; they are sufficient in number to ensure Liouville integrability.

\end{example}

\subsection{Integrability}\label{par:integrability}
We show in this subsection that if an LV system $\LV(\Gamma)$ is Liouville integrable, or superintegrable, then
also any other LV system $\LV(\Gamma')$ having the same decloning, i.e., for which $\ul\Gamma=\ul{\Gamma'}$. In
particular, $\LV(\Gamma)$ is Liouville integrable, or superintegrable, if and only if the decloned system
$\LV(\ul{\Gamma})$ is Liouville integrable, or superintegrable. Recall from Section~\ref{par:examples} that if
$\Gamma=(S,A)$ is a skew-symmetric graph of order $n$, then $\LV(\Gamma)$ is $n$-dimensional and for the Liouville
integrability of $\LV(\Gamma)$, we need $n-\frac12\Rk{A}$ independent functions in involution, among which the
Hamiltonian $H_{S}$; for the superintegrability of $\LV(\Gamma)$, we need $n-1$ independent first integrals. Let
$\wght$ be a weight vector for $\Gamma$. Then $\LV(\Gamma^\wght)$ has dimension $\vert\wght\vert$, so that in
particular, for the
\begin{enumerate}
  \item[$\bullet$] Liouville integrability of $\LV(\Gamma^\wght)$, we need $\vert\wght\vert-\frac12\Rk{A}$
    independent functions in involution, among which the Hamiltonian $H_{S^\wght}$;
  \item[$\bullet$] Superintegrability of $\LV(\Gamma^\wght)$, we need $\vert\wght\vert-1$ independent first
    integrals.
\end{enumerate}

\begin{prop}\label{prop:int}%
  Let $\Gamma$ and $\Gamma'$ be two skew-symmetric graphs with the same decloned graph, i.e.
  $\ul{\Gamma}=\ul{\Gamma'}$. If the LV system $\LV(\Gamma)$ is Liouville integrable, or superintegrable,
  then the same is true for  $\LV(\Gamma')$.
\end{prop}
\begin{proof}
It is clear that is enough to show that if $\Gamma$ is irreducible then $\LV(\Gamma)$ is Liouville integrable
(or superintegrable) if and only if $\LV(\Gamma^\wght)$ is Liouville integrable (resp. superintegrable),
where $\wght$ is any weight vector for $\Gamma$.

Let us denote, as before, the Poisson structures of $\LV(\Gamma)$ and of $\LV(\Gamma^\wght)$ respectively by
$\pi_A$ and $\pi_{A^\wght}$; the standard coordinates on $\bbF[S]$ and on $\bbF\[S^\wght\]$ are denoted
respectively by $x_s$ and $x_{s_i}$, with $s\in S$ and $i\in\set{1,2,\dots,\wght(s)}$; also, $n$ stands for the
order of $\Gamma$. We recall from Section \ref{par:LV_cloning} that $\pi_A$ and $\pi_{A^\wght}$ have the same rank,
equal to the rank of the adjacency matrix $A$ of $\Gamma$; let us denote this even integer by $2r$.  Recall also
from that section that we have a first set of $\vert\wght\vert-n$ independent Casimir functions for
$\pi_{A^\wght}$, which are of the form $x_{s_i}/x_{s_1}$, for $s\in S$ and $i=2,3,\dots,\wght(s)$.  We consider new
coordinates~$y_{s_i}$, with $s\in S$ and $i\in\set{1,2,\dots,\wght(s)}$, on an open dense subset of
$\bbF\[S^\wght\]$; they are defined by:
\begin{align*}\label{y-coord}
y_{s_1}&:=\LV(p)^*x_s=\sum_{i=1}^{\varpi{(s)}} x_{s_i}\;,\quad\text{for } s\in S\;,\\
y_{s_j}&:=x_{s_j}/x_{s_1}\;,\quad\text{for } s \in S, \text{ and }  j\in\{2,3,\dots,\wght(s)\}\;,
\end{align*}
where $\LV(p):\bbF\[S^\wght\]\rightarrow \bbF[S]$ is the decloning map, defined in Proposition~\ref{prp:decloning}.
It is clear from \eqref{eq:LV_wght} that, in the new coordinates, $\LV(\Gamma^\wght)$ decouples into a subsystem,
isomorphic to $\LV(\Gamma)$, and a trivial system, i.e., a system with no dynamics:  explicitly, the decoupled
system reads
\begin{equation}\label{eq:dec_syst}
\dot y_{s_1}=y_{s_1}\sum_{t\in S}a_{s,t}y_{t_1},\quad \dot y_{s_j}=0,\quad
s \in S, \;\;  j\in\set{2,3,\dots,\wght(s)}\;.
\end{equation}

Suppose that $\LV(\Gamma)$ is Liouville integrable, respectively superintegrable, with independent first integrals
$F_1, F_2,\dots, F_\ell$. On $\bbF\[S^\wght\]$ we consider the functions $\LV(p)^*F_1,\dots, \LV(p)^*F_\ell$, as
well as the Casimir functions $y_{s_i}$ for $s\in S$ and $i\in\set{2,3,\dots,\wght(s)}$. The former functions are
first integrals for $\LV(\Gamma^\wght)$ because $\LV(p)$ is a Poisson map, with $\LV(p)^*H_S=H_{S^\wght}$; for the
same reason, they are in involution when the functions $F_1,\dots,F_\ell$ are in involution. The functions
$\LV(p)^*F_1,\dots, \LV(p)^*F_\ell$ are independent because $\LV(p)$ is a submersion; since they only depend on the
variables $y_{s_1}$, with $s\in S$, they are also independent of the Casimirs $y_{s_i}$, with $i>1$. It follows
that we have enough independent functions in involution for the Liouville integrability, respectively enough
independent functions for the superintegrability of $\LV(\Gamma^\wght)$.

In order to prove the inverse implication, suppose that $\LV(\Gamma^\wght)$ is Liouville integrable, respectively
superintegrable, with first integrals $G_1, G_2,\dots, G_\ell$ complemented with the Casimir functions
$\set{y_{s_i}:s\in{S},i\in\set{2,3,\dots, \wght(s)}}$. If we fix these Casimir functions to any specific values
$c_{s_i}\in\bbF$, then the restricted functions $G_1',\dots,G_r'$ depend only on $y_{s_1}$, with $s\in S$, hence
are pullbacks under the Poisson submersion $\LV(p)$ of some functions $F_1,\dots,F_\ell$, defined on an open subset
of $\bbF[S]$. Since these restricted functions are also first integrals, it follows as before that the functions
$F_1,\dots,F_\ell$ are first integrals of $\LV(\Gamma)$, and that they are in involution when the functions
$G_1,\dots,G_\ell$ are in involution. Moreover, for generic values of the constants $c_{s_i}$, they are
independent. It follows that $\LV(\Gamma)$ is Liouville integrable, respectively superintegrable.
\end{proof}

Since the decloning map $\LV(p)$ is a Poisson map, the construction of the integrals in the above proof shows that
if $\Gamma$ is a skew-symmetric graph for which $\LV(\Gamma)$ is Liouville or superintegrable, then
$\LV(p):\LV(\Gamma)\to\LV(\ul{\Gamma})$ is a morphism of integrable systems.

\subsection{Lax equations}\label{subsec:lax}
In this section we show how a Lax equation for an LV system $\LV(\Gamma)$ leads to a Lax equation for any of the
cloned systems $\LV(\Gamma^\wght)$.  In the case of $\LV(\Gamma)=\B nk$, we will also provide an alternative Lax
equation.

Suppose that $\dot L=[L,M]$ is a Lax equation for the LV system $\LV(\Gamma)$, where $\Gamma=(S,A)$ is any
skew-symmetric graph of order $n$. In order to indicate explicitly how $L$ and $M$ depend on the $x$-variables, we
write the matrices $L$ and $M$ also as $L(x)$ and $M(x)$. We construct a Lax equation for $\LV(\Gamma^\wght)$,
where $\wght$ is any weight vector for $\Gamma$. To do this, we use the decloning map
$\chi:=\LV(p):\bbF\[S^\wght\]\rightarrow\bbF\[S\]$, defined in Proposition~\ref{prp:decloning}. We also let
$y_s:=\chi^*x_s=\sum_{i=1}^{\wght(s)}x_{s_i}$ for all $s\in S$. Then the vector field \eqref{eq:LV_wght} of
$\LV(\Gamma^\wght)$ can be written in the compact form
\begin{equation}\label{eq:LV_wght_lax}
  \dot{x}_{s_i} = x_{s_i}\sum_{t\in S}a_{s,t}y_t\;, \qquad \text{for}\; s\in S
  \text{ and } i=1,\dots,\wght(s)\;.
\end{equation}
Summing up these equations for $i=1,\dots,\wght(s)$, we get
\begin{equation*}
  \dot y_s=\sum_{i=1}^{\wght(s)}\dot x_{s_i}=\sum_{i=1}^{\wght(s)}x_{s_i}\sum_{t\in S}a_{s,t}y_t=
  y_{s}\sum_{t\in S}a_{s,t}y_t\;,
\end{equation*}
which has exactly the same form as \eqref{eq:LV_System}, for which we have the Lax equation
$\dot{L}(x)=[L(x),M(x)]$. It follows that $\dot L(y)=[L(y),M(y)]$ is a Lax equation for
$\LV(\Gamma^\wght)$. Suppose that the former Lax equation is regular, which means that the dynamics is completely
determined by it. Then the latter Lax equation is also regular, because the Poisson manifold
$(\bbF\[S^\wght\],\pi_{A^\wght})$ has Casimirs functions $x_{s_i}/x_{s_1}$, for $1\leqslant s\leqslant n$ and
$i=2,3,\dots,\wght(s)$. Also, when the former Lax equation depends on a parameter, then so does the latter Lax
equation. Finally, if $F$ is a first integral of $\LV(\Gamma)$ which appears as a coefficient of the characteristic
polynomial of $L(x)$, then $\chi^*F$ is a first integral of $\LV(\Gamma^\wght)$ which appears as a coefficient of
the characteristic polynomial of~$L(y)$; in particular, if the Lax operator $L(x)$ provides enough first integrals
to prove the Liouville integrability of $\LV(\Gamma)$ then so does $L(y)$ for~$\LV(\Gamma^\wght)$.

We now turn to the particular case of $\B nk$. Let $k$ and $n$ be fixed, with $k<n/2$ and denote as in Example
\ref{exa:bogo} by $\Gamma_{n,k}=(S_n,A_{n,k})$ the graph underlying the LV system $\B nk$. Bogoyavlenskij
\cite{Bog2} constructed for $\B nk$ the following Lax equation with spectral parameter $\lambda$:
\begin{equation}\label{eq:bogo_lax}
  \dot{L}=\lb{L,M}\;,\ \text{where}\ \left\{
  \begin{array}{lll}
    L:=L_0+\l \Sh\;, \ &M:=M_0-\l \Sh^{k+1}\;,\\
    L_0:=X\,\Sh^{-k}\;,&M_0:=\sum_{t=k+1}^{n-1}\Sh^{t}\,X\,\Sh^{-t}\;,
  \end{array}\right.
\end{equation}
and where $\Sh$ is the circular shift matrix defined by $\Sh_{i,j}:=\delta_{i+1,j}$, and~$X$ is the diagonal matrix
$\hbox{diag}(x_1,x_2,\dots,x_n)$; notice that $M_0$ is also a diagonal matrix.  We fix a weight vector $\wght$ for
$\Gamma_{n,k}$.  It follows from what precedes that $\dot L(y)=\lb{L(y),M(y)}$ is a Lax equation for the cloned
system $\LV(\Gamma_{n,k}^\wght)$, where we recall that $y_i=\chi^*x_i$ for $i=1,\dots,n$; said differently, the
latter Lax equation is obtained by replacing in \eqref{eq:bogo_lax} everywhere the diagonal matrix $X$ by the
diagonal matrix $\chi^*X$.

In what follows we provide a different Lax equation $\dot\cL=\lb{\cL,\cM}$ for the cloned system
$\LV(\Gamma_{n,k}^\wght)$. The new Lax equation, which is also regular and also depends on a spectral parameter,
has the advantage that all phase variables $x_{s_i}$ of the cloned system appear as entries of the new Lax
operator~$\cL$. The matrices $\cL$ and $\cM$ will be block matrices with $N\times N$ square blocks of order $n$,
where $N$ is the maximum number of clones of a vertex, $N:=\max\set{\varpi(s)\mid s\in S_n}$; for $1\leqslant
i,j\leqslant N$ the~$(i,j)$-block of such a block matrix $\cN$ is denoted by $\cN_{(i,j)}$. We let $x_{s_i}=0$ for
$\varpi(s)<i\leqslant N$ and $s\in S_n$, so that every vertex of $\Gamma$ has now the same number of clones. With
this notation, the cloned Bogoyavlenskij system $\LV(\Gamma_{n,k}^\wght)$ can be written in the simpler form
\begin{equation}\label{eq:Bogo_wght}
  \dot{x}_{s_i} = x_{s_i}\sum_{t=1}^k(y_{s+t}-y_{s-t})\;, \quad \text{for}\; 1\leqslant s\leqslant n
  \text{ and } 1\leqslant i\leqslant N\;,
\end{equation}
where the indices of the variables $y_s$ are again taken modulo $n$ and in $S_n$.
Let us denote, for $1\leqslant i\leqslant N$, by $X^{(i)}$ the diagonal matrix
$\hbox{diag}(x_{1_i},x_{2_i},\dots,x_{n_i})$, and in analogy with (\ref{eq:bogo_lax}), let
\begin{equation}\label{eq:bogo_lax_2}
  \begin{array}{lll}
    L^{(i)}:=L_0^{(i)}+\l\Sh\;,\ &M^{(i)}:=M_0^{(i)}-\l\Sh^{k+1}\;,\\
    L_0^{(i)}:=X^{(i)}\,\Sh^{-k}\;,&M_0^{(i)}:=\sum_{t=k+1}^{n-1}\Sh^{t}\,X^{(i)}\,\Sh^{-t}\;.
  \end{array}
\end{equation}
We first show that, in terms of these matrices, \eqref{eq:Bogo_wght} can be written as the following equation of
Lax type:
\begin{equation}\label{eq:Lax_to_prove}
  \dot L_0^{(i)}=\lb{L_0^{(i)},\chi^*M_0}=\sum_{r=1}^N\lb{L_0^{(i)},M_0^{(r)}}\;.
\end{equation}%
To do this, we first rewrite the first equality of \eqref{eq:Lax_to_prove}, using \eqref{eq:bogo_lax} and
\eqref{eq:bogo_lax_2}; we also use that $X^{(i)}$ and $\Delta^tY\Delta^{-t}$ are diagonal matrices, hence commute,
where $Y:=\chi^*X$:
\begin{align*}
  \dot X^{(i)}&=\sum_{t=k+1}^{n-1}X^{(i)}\Delta^{t-k}Y\Delta^{k-t}-\sum_{t=k+1}^{n-1}\Delta^{t}Y\Delta^{-t}X^{(i)}\\
  &=X^{(i)}\(\sum_{t=k+1}^{n-1}\Delta^{t-k}Y\Delta^{k-t}-\sum_{t=k+1}^{n-1}\Delta^{t}Y\Delta^{-t}\)\\
  &=X^{(i)}\sum_{t=1}^{k}\(\Delta^{t}Y\Delta^{-t}-\Delta^{n-t}Y\Delta^{t-n}\)\;.
\end{align*}%
The equivalence with \eqref{eq:Bogo_wght} then follows by taking in both sides of the latter equality, which are
diagonal matrices, the $(s,s)$-th entry.

%


\begin{prop}\label{prop:Large_Lax_from_small_Lax}
  Let $\cL$ and $\cM$ be the square matrices of order $nN$, whose blocks are defined for $1\leqslant i,j\leqslant
  N$ by $\cL_{(i,j)}:=L^{(j)}$ and
  $$
  \cM_{(i,j)}:=\delta_{i,j}\sum_{r=1}^N\left(M_0^{(r)}+\Sh^k\,X^{(r)}\,
               \Sh^{-k}\right)-\Sh^k\,X^{(j)}\,\Sh^{-k}-\l\Sh^{k+1}\;.
  $$
  Then $\dot{\cL}=\lb{\cL,\cM}$ is a Lax equation with  spectral parameter for the cloned Bogoyavlenskij system
  $\LV(\Gamma^\varpi_{n,k})$.
\end{prop}

\begin{proof}
For $1\leqslant i,j\leqslant N$ we show that the $(i,j)$-blocks of both sides of $\dot{\cL}=\lb{\cL,\cM}$ agree. On
the one hand, we get by using \eqref{eq:Lax_to_prove},
\begin{equation}\label{eq:Lax_left}
  \dot\cL_{(i,j)}=\dot L^{(j)}={\dot L_0}^{(j)}=\sum_{r=1}^N\lb{L_0^{(j)},M_0^{(r)}}\;.
\end{equation}%
The $(i,j)$-block of the commutator $\lb{\cL,\cM}=\cL\cM-\cM\cL$ is by block multiplication given by
\begin{align*}
\sum_{r=1}^N&\left(\Sh^k\,X^{(r)}\,\Sh^{-k}+\l \Sh^{k+1}\right)\,L^{(j)}-
\sum_{r=1}^NL^{(r)}\,\left(\Sh^k\,X^{(j)}\,\Sh^{-k}+\l \Sh^{k+1}\right)+\\
\sum_{r=1}^N&\lb{L^{(j)},M_0^{(r)}+\Sh^k\,X^{(r)}\,\Sh^{-k}}
=\sum_{r=1}^N\left(A_0(r)+\l\,A_1(r)+\l^2\,A_2(r)\right)\,,
\end{align*}
with
\begin{align*}
A_0(r)&=L_0^{(j)}\,\Sh^{k}\,X^{(r)}\,\Sh^{-k}-L_0^{(r)}\,\Sh^{k}\,X^{(j)}\,\Sh^{-k}+
L_0^{(j)}\,M_0^{(r)}-M_0^{(r)}\,L_0^{(j)}\,,\\
A_1(r)&=\Sh^{k+1}\,L_0^{(j)}+\Sh^{k+1}\,\left(X^{(r)}-X^{(j)}\right)\,\Sh^{-k}-L_0^{(r)}\,\Sh^{k+1}+
\lb{\Sh,M_0^{(r)}}\,,\\
A_2(r)&=\Sh^{k+2}-\Sh^{k+2}=0\,.
\end{align*}
In order to compute $A_1(r)$, first notice that
\begin{align*}
  \lb{\Sh,M_0^{(r)}}&=
    \sum_{t=k+1}^{n-1}\left(\Sh^{t+1}\,X^{(r)}\,\Sh^{-t}-\Sh^{t}\,X^{(r)}\,\Sh^{1-t}\right)\\
   &=X^{(r)}\,\Sh-\Sh^{k+1}\,X^{(r)}\,\Sh^{-k}=\lb{L_0^{(r)},\Sh^{k+1}}\;,
\end{align*}
where we have used the definitions (\ref{eq:bogo_lax_2}) and $\Sh^n=\Id_n$. It follows, using again the definition
of $L_0^{(j)}$, that
\begin{align*}
  A_1(r)&=\Sh^{k+1}\,X^{(r)}\,\Sh^{-k}-L_0^{(r)}\,\Sh^{k+1}+\lb{\Sh,M_0^{(r)}}\\
        &=\Sh^{k+1}\,X^{(r)}\,\Sh^{-k}-\Sh^{k+1}\,L_0^{(r)}=0\,.
\end{align*}
Using again the definition of $L_0^{(j)}$, the first two terms of $A_0(r)$ cancel, so that
$A_0(r)=\lb{L_0^{(j)},M_0^{(r)}}$. Summing up, we get
$$
  \lb{\cL,\cM}_{(i,j)}=\sum_{r=1}^N A_0(r)=\sum_{r=1}^N\lb{L_0^{(j)},M_0^{(r)}}\,,
$$
so that, using \eqref{eq:Lax_left}, $\dot\cL_{(i,j)}=\lb{\cL,\cM}_{(i,j)}$, as was to be shown.
\end{proof}

\end{document}